\documentclass[lettersize, journal]{IEEEtran}
\usepackage{amssymb}
\usepackage{amsmath}
\usepackage{cite}
\usepackage{url}
\usepackage{stfloats}
\usepackage{booktabs}
\usepackage{xcolor}
\usepackage{cite,graphicx,amsmath,amssymb}
\usepackage{subfigure}
\usepackage{fancyhdr}

\usepackage{mdwmath}
\usepackage{mdwtab}
\usepackage{caption}
\usepackage{amsthm}
\usepackage{setspace}
\usepackage{algorithm}
\usepackage{algorithmic}
\usepackage{makecell}
\usepackage{diagbox}
\usepackage{amsmath}

\newtheorem{theorem}{Theorem}

\newtheorem{lemma}{Lemma}

\newtheorem{corollary}{Corollary}

\allowdisplaybreaks
\makeatletter
\def\ScaleIfNeeded{%
	\ifdim\Gin@nat@width>\linewidth \linewidth \else \Gin@nat@width
	\fi } \makeatother
 
\begin{document}
	\title{Movable-Element STARS-Assisted NOMA Communications}
	
	\author{Jingjing Zhao,~\IEEEmembership{Senior Member,~IEEE}, Yiheng Zhang, Kaiquan Cai, Xidong Mu, and Yuanwei Liu,~\IEEEmembership{Fellow,~IEEE}
		
		\thanks{J. Zhao, Y. Zhang, K. Cai are with the School of Electronics and Information Engineering, Beihang University, 100191, Beijing, China, and also with the State Key Laboratory of CNS/ATM, 100191, Beijing, China. (e-mail:\{jingjingzhao, zhangyiheng, ckq\}@buaa.edu.cn). 
			
			X. Mu is with the Centre for Wireless Innovation (CWI), Queen's University Belfast, Belfast, BT3 9DT, U.K. (e-mail: x.mu@qub.ac.uk). 
			
			Y. Liu is with the Department of Electrical and Electronic Engineering, the University of Hong Kong, Hong Kong, China (e-mail: yuanwei@hku.hk). }}
	
	\maketitle
	
	\vspace{-1.5cm}
	\begin{abstract}
		This paper proposes a movable-element (ME) enabled simultaneously transmitting and reflecting surface (ME-STARS)-assisted non-orthogonal multiple access (NOMA) communication system, where the positions of STARS elements can be flexibly reconfigured to introduce additional spatial degrees-of-freedom. Our objective is to maximize the sum achievable rate by jointly optimizing the active beamforming at the base station (BS) as well as the passive beamforming and MEs positions at the STARS. To tackle the resultant highly-coupled non-convex optimization problem, an alternating optimization-based algorithm is proposed, in which the original problem is decomposed into three subproblems. For the BS and ME-STARS beamforming subproblems, the successive convex approximation method is  employed. For the MEs positions optimization, penalty terms are incorporated into the objective function to handle the constraints, thereby transforming the original problem into an unconstrained formulation. The gradient descent algorithm is subsequently applied to iteratively approach the sub-optimal solution. Numerical results demonstrate that: i) ME-STARS-assisted NOMA communications achieve superior performance compared to the NOMA communications assisted by STARS with fixed-position elements; and ii) the performance gain brought by the MEs is significant even with restricted moving region size.
		
	\end{abstract}
	
	\begin{IEEEkeywords}
		Beamforming, movable elements, non-orthogonal multiple
		access, simultaneously
		transmitting and reflecting surface.
	\end{IEEEkeywords}
	\section{Introduction}
	The exponential growth in data traffic over wireless communication systems has resulted in a pronounced scarcity of available spectral resources, prompting extensive research for improving the spectral efficiency (SE). Reconfigurable intelligent surfaces (RISs) have garnered significant research interest owing to the unique capability to enhance SE by dynamically reconfiguring wireless propagation environments~\cite{RIS1, RIS2}. Specifically, the fundamental architecture of RISs consists of a planar surface populated with numerous low-cost passive elements designed to dynamically reconfigure both the amplitudes and phase shifts of incoming electromagnetic waves. Traditional reflective-only RISs require both the transmitter and receiver to be co-located within the same hemispherical region relative to the RISs, which imposes constraints on the adaptability of practical application scenarios. To address this coverage constraint, the simultaneously transmitting and reflecting surface (STARS) has been developed as an effective solution through synchronous manipulation of incoming electromagnetic waves for both transmission and reflection to achieve full-space coverage~\cite{STARS1, STARS2}.
	
	Note that current RIS/STARS are predominantly deployed with fixed position antennas (FPAs), which inherently imposes restrictions on their ability to exploit continuous spatial channel variations. The antenna selection (AS) technique~\cite{AS} can exploit spatial degrees-of-freedom (DoFs) for the FPA systems to a certain extent via choosing the antennas providing better channel conditions. However, with the increment of the utilization of spatial diversity, the scale of antenna arrays becomes massive, leading to a significant increase in hardware expenditures. To overcome this limitation, recent studies have explored position-adjustable antennas (PAA) technologies, including fluid antenna (FA)~\cite{FA1,FA2} and movable antenna (MA)~\cite{MA1,MA2}, to achieve enhanced system performance through dynamic spatial optimization at low cost. Specifically, the positions of antennas can be dynamically optimized within a spatially constrained region spanning multiple wavelengths, enabling precise channel condition enhancement. Building upon these demonstrated advantages, recent research efforts have progressively shifted focus toward investigating the movable-element (ME) enabled RIS (ME-RIS) communications~\cite{ME-RIS2} and ME enabled STARS (ME-STARS) communications~\cite{ME-STAR2}. 
	
	Meanwhile, the non-orthogonal multiple access (NOMA) has emerged for enabling users to share the same resource blocks, thereby supporting massive connectivity~\cite{NOMA1,NOMA2}. With superposition coding (SC) and successive interference cancellation (SIC), the intra-channel interference can be effectively mitigated by distinguishing users' signals in the power domain~\cite{NOMA3}. Since the SIC decoding order highly depends on the channel conditions of users, the ME-STARS can improve the design flexibility of NOMA networks with reconfigured channel qualities by adjusting both the transmission/reflection coefficients and MEs positions. Thus, the ME-STARS-assisted NOMA scheme constitutes a foundational framework for the evolution of wireless communications toward 6G and beyond.
	
	\subsection{Related Works}
	\subsubsection{Studies on RIS/STARS-assisted NOMA communications}
	Existing researches have explored RIS-assisted NOMA systems to leverage reconfigurable propagation environments for channel quality enhancement. In~\cite{RIS-NOMA3}, the authors investigated single-input single-output RIS-aided NOMA systems, where the phase shifts, subchannel assignment, and power allocation were jointly optimized. For multi-cell RIS-aided NOMA networks, a joint optimization framework was developed in~\cite{RIS-NOMA5}, where successive convex approximation (SCA) and swap-matching-based algorithms were employed to jointly design user pairing, subchannel assignment, transmit power allocation, RIS beamforming, and the SIC decoding order. Considering the scenario for the multiple-input-single-output communications, the authors of~\cite{RIS-NOMA6} investigated the energy-efficiency (EE) maximization problem, solving the non-convex optimization problem via an algorithmic framework comprising semidefinite programming (SDP) and SCA.  Considering the existence of eavesdropping, a physical-layer security framework was studied in~\cite{RIS-NOMA8} for RIS-aided NOMA systems, where a max-min secrecy rate optimization problem was formulated for both internal and external eavesdropping scenarios. Furthermore, a sum-rate maximization framework was proposed in~\cite{STAR-RIS-NOMA1} for STARS-assisted NOMA system, where the passive beamforming at the STARS as well as the active beamforming and power allocation at the base station (BS) were jointly optimized. In~\cite{STAR-RIS-NOMA2}, both cluster-based and beamformer-based NOMA were considered for STARS-aided communications, and the corresponding weighted sum-rate maximization problem was investigated. In~\cite{STAR-RIS-NOMA6}, the authors investigated the EE maximization problem in a STARS-assisted NOMA downlink system by jointly optimizing BS beamforming and STARS beamforming, where the fractional programming method was adopted. A fairness-enhanced integrated sensing and communication (ISAC) system leveraging STARS and NOMA was proposed in~\cite{STAR-RIS-NOMA5}, where the joint design of beamforming and STARS coefficients was conducted to ensure equitable performance across communication terminals and the sensing target. Under imperfect channel state information (CSI), the authors of~\cite{CSI-RIS-NOMA} investigated secure transmissions in the RIS-NOMA system, where a robust beamforming scheme leveraging the artificial noise was proposed to minimize the transmit power.
	
	\subsubsection{Studies on PAA-assisted communications}
	Current PAA technologies mainly include FA and MA. Specifically, the FA technology achieves dynamic adjustment of antenna positions by changing the shape or position of liquid metal or electrolyte solutions. The FA technology was first introduced in~\cite{FA1}, where the antennas were enabled to be repositioned among predefined candidate ports along a constrained linear domain to maximize the channel gain. In~\cite{FA2}, the authors investigated the FA-aided multi-user system, optimally positioning each user's antenna to exploit channel fading characteristics for concurrent interference suppression and desired signal enhancement. Different from FAs, which can only achieve one-dimensional discrete space movement due to the limitations of liquid materials, the MAs driven by motors can achieve two-dimensional (2D) continuous space movement.  In~\cite{MA1}, the authors established a pioneering field-response channel model for MA-aided communications and derived the theoretical bounds for achievable channel gain under both deterministic and stochastic channel conditions. Furthermore, for the multi-user system with single-MA users, the authors of~\cite{MA2} quantified multiuser access gain through joint optimization of antenna positioning vectors, users transmit power, and BS combiner matrices. Given the joint dependence of MIMO beam patterns on both beamforming weights and antenna positions, the authors of~\cite{MA3} leveraged the antennas positions optimization to maximize the channel capacity. In~\cite{MA4},  the authors studied a MA-assisted NOMA downlink system, where the channel capacity was maximized by jointly designing user-side MAs positions and transmit power allocations via an alternating optimization (AO) algorithm.
	
	Recently, growing attention is directed to investigate the combination of RIS/STARS and PAA technologies. In~\cite{ME-RIS2D}, the authors demonstrated that MA-enabled RISs yields better performance over BS deployed with MAs, where a Riemannian product manifold optimization framework was employed to optimize the coupled transmit beamforming and the positions of MAs. In~\cite{11033708}, to address the critical security issue of secrecy leakage in ISAC system, the authors proposed to leverage the synergistic benefits of MAs and RIS for physical layer security enhancement, where a penalty-based algorithm was developed to optimize the beamformers, RIS coefficients, and MA positions against an eavesdropping target. Moreover, considering the joint effects of total transmit power and the number of elements, the authors of~\cite{ME-RIS3} quantitatively analyzed the outage probability performance of MA-enabled RIS architectures. To further enhance the spatial DoFs at the RIS side, the authors of~\cite{ME-RIS2} deployed MEs at the RIS, where the RIS elements positions could be flexibly adjusted and a low-complexity algorithm was proposed to eliminate the phase distribution offset across channels. Moreover, a ME-STARS-assisted near-field wideband communication framework was proposed in~\cite{ME-STAR2}, where a two-layer algorithm was invoked for jointly optimizing the BS precoding as well as the passive beamforming and MEs positions to maximize the sum rate. 
	\subsection{Motivations and Contributions}
	
	Despite the aforementioned research efforts in STARS and PAA technologies, the integration of ME-STARS with NOMA communications remains an under-explored research domain. On the one hand, the joint optimization of the BS active beamforming, the ME-STARS passive beamforming, and the MEs positions is challenging due to the highly coupled variables. On the other hand, the reconfigurable propagation environment brought by adjustable transmitting/reflecting coefficients and MEs positions introduces new difficulties to the application of SIC for NOMA.
	
	To overcome the above challenge, this paper explores the ME-STARS-assisted NOMA communications and study the joint beamforming and MEs positions optimization problem with the primary objective of achieving sum rate maximization. The main contributions are summarized as follows:
	\begin{itemize}
		\item We propose a ME-STARS-assisted NOMA communication system, where MEs are deployed on the STARS to harness the inherent spatial-domain characteristics of the wireless channel. Under this framework, we formulate a joint BS beamforming, ME-STARS beamforming, and MEs positions design problem to maximize the sum achievable rate, subject to the NOMA SIC decoding constraint.
		\item We propose an AO algorithm to solve the resultant problem with highly coupled variables, where three subproblems decomposed from the original problem are alternatively solved in each iteration. The SCA technique is applied to effectively solve the BS beamforming optimization subproblem. Besides, we develop a penalty-based method to deal with the rank-one constraint in the ME-STARS beamforming subproblem. For the design of MEs positions, we adopt the penalty method and Sigmoid function to transform the constrained problem into an unconstrained one, which is followed by the gradient descent algorithm (GDA) for obtaining the suboptimal solution. 
		\item Simulation results reveal that: 1) the ME-STARS-assisted NOMA communications enables considerable sum rate improvement compared to the NOMA communications assisted by STARS with fixed-position elements (FPE-STARS); 2) the performance merit of ME-STARS remains significant even under restricted moving region size.
	\end{itemize}
	\subsection{Organization and Notations}
    Section $\text{II}$ establishes the system model and formulates the associated sum-rate maximization problem. Building on this formulation, Section $\text{III}$ develops an iterative AO-based algorithm for the BS beamforming, ME-STARS beamforming, and MEs positions optimization problems. Numerical results and comparisons with benchmark schemes are presented in Section \text{IV} to evaluate the proposed algorithm. The main conclusions are summarized in Section \text{V}.
	
	\textit{Notations}: The spaces of $N\times M$ complex- and real-valued matrices are represented by $\mathbb{C}^{N\times M}$ and $\mathbb{R}^{N\times M}$, respectively. For vector $\mathbf{s}$, $\mathbf{s}^{*}$, $\mathbf{s}^{T}$, and $\mathbf{s}^{H}$ represent its conjugate,  transpose, and conjugate transpose, respectively. The Euclidean norm of $\mathbf{s}$ and the Frobenius norm of a matrix $\mathbf{Q}$ are written as $\|\mathbf{s}\|_2$ and $\|\mathbf{Q}\|_F$, while $\operatorname{Tr}(\mathbf{Q})$ represents the trace of $\mathbf{Q}$. Moreover, $\mathbf{1}_L$ refers to the $L$-dimensional all-ones vector, and $\circ$ represents the Hadamard product.

	\section{System Model and Problem Formulation}
	In this section, the system model of the ME-STARS-assisted NOMA communication system is first introduced. Then, we formulate the joint BS beamforming, the ME-STARS beamforming, and the MEs positions optimization problem.
	\subsection{System Model}
	\begin{figure}[t]
		\centering
		\includegraphics[width=3.6in]{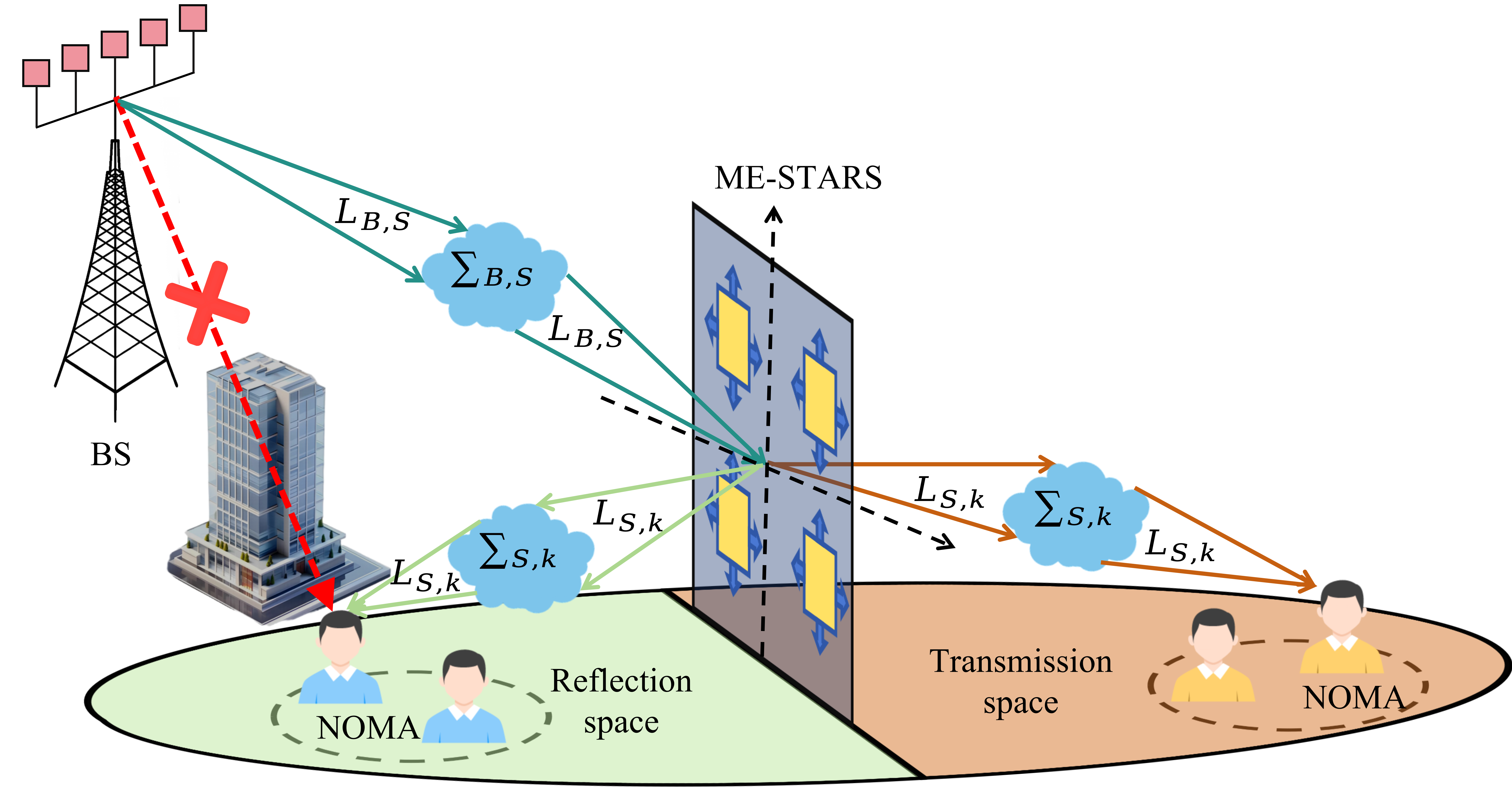}
		\caption{Illustration of the ME-STARS assisted NOMA communication system. }
		\label{System model}
	\end{figure}
	As shown in Fig. 1, we consider a downlink NOMA communication system, where $K$ users receive data transmitted from BS with $M$ antennas through signal transmissions and reflections facilitated by a ME-STARS composed of $N$ MEs\footnote{The STARS MEs can be flexibly adjusted along slide tracks using step motors \cite{motor}.}. As for users, denoted by $k\in\mathcal{K} = \left\{1, \ldots, K\right\}$, each is equipped with a single antenna. Considering a three-dimensional (3D) Cartesian coordinate system, the locations of the BS and user $k$ are represented by $\mathbf{b}_0=\left[x_b,y_b,z_b\right]$ and $\mathbf{u}_{k} = \begin{bmatrix} x_{k}, y_{k},0\end{bmatrix}^T$, respectively. Specifically, in the considered system, users are categorized into two groups according to their positions relative to the ME-STARS. The group located in the region where the ME-STARS reflects signals is denoted by the set $\mathcal{R}$, while the group located in the region where the ME-STARS transmits signals is denoted by the set $\mathcal{T}$. We assume that the STARS array has a planar shape of size $ N = N_{x} \times N_{y}$, with $\mathcal{N}=\left\{1,2,\cdots,N\right\}$ denoting the set of the MEs. 
	The position of the ME-STARS is denoted by $\mathbf{q}_{c} = \begin{bmatrix} x_c, y_c ,z_c \end{bmatrix}^T$. The local coordinate systems at the BS and STARS are denoted by $x_b-O_b-y_b$ and $x_s-O_s-y_s$, with $O_b=\left[0,0\right]^T$ and $O_s=\left[0,0\right]^T$ denoting the local original points, respectively. Then, the local coordinates for the $m\text{-}$th FPA and $n\text{-}$th ME are given by the vectors $\mathbf{s}_m=\left[x_m, y_m\right]^T$ and $\mathbf{q}_{n} = \begin{bmatrix} x_{n}, y_{n} \end{bmatrix}^T $, respectively. We further denote the collections of the coordinates of $N$ MEs by $\mathbf{Q}=\begin{bmatrix}
		\mathbf{q}_1,\mathbf{q}_2,\ldots,\mathbf{q}_N 
	\end{bmatrix}\in \mathbb{R}^{2 \times N}$. The movable region constraint of MEs is given by $\mathbf{q}_{n} \in \mathcal{C}_{\rm ME}=\left \{\begin{bmatrix} x, y\end{bmatrix}^T \mid x \in \left[0, A\right], y \in \left[0, A\right]\right\}$, with $\mathcal{C}_{\rm ME}$ indicating the feasible movement region of the MEs. It is noted that, due to the limited moving speed of MEs, the proposed ME-STARS-assisted NOMA system is particularly suitable for communication scenarios where the wireless channels vary slowly over time, such as Internet of Things.
	
	Without loss of generality, we assume that the ME-STARS operates in the energy splitting (ES) mode~\cite{proof}. 
	Specifically, the transmission- and reflection-coefficient matrices of the MEs are expressed as $\mathbf{\Theta}_t={\text{diag}}\left(\mathbf{u}_t^H\right)\in\mathbb{C}^{N \times N}$ and $\mathbf{\Theta}_r={\rm{diag}}\left(\mathbf{u}_r^H\right)\in\mathbb{C}^{N \times N}$, respectively, with $\mathbf{u}_s=\left[u^s_1, u^s_2,\cdots,u^s_N\right]^H\in\mathbb{C}^{N \times 1}$ and $u^s_n=\sqrt{\beta_n^s}e^{j\theta_{n}^s}$, $\forall s\in\{r,t\}$.
	Here, $\beta_n^t,\beta_n^r\in[0,1]$ and $\theta_n^t, \theta_n^r \in [0,2\pi), \forall n\in\mathcal{N}$ denote the amplitude and phase-shift\footnote{It is worth noting that the continuous phase-shift model is adopted to explore the maximum performance of the proposed ME-STARS-NOMA framework. 
		} coefficients of the $n$-th element, respectively. According to the energy conservation, we have $\beta_n^t+\beta_n^r=1, \forall n\in\mathcal{N}$. 

	\subsection{Signal Model}
    The direct links between the BS and users are assumed to be obstructed by the physical barriers in the dense urban environments. Both the BS-ME-STARS and the ME-STARS-users channels are dominated by line-of-sight (LoS) links, where there exist $L_{\text{B,S}}$ and $L_{\text{S},k}$ paths respectively. Let $\theta_{\text{B}}^e$ and $\varphi_{\text{B}}^e$  represent the elevation
	and azimuth angle of departure (AoD) of the $e$-th path from the BS to the ME-STARS, then the field-response vector (FRV) of the $m\text{-}$th BS antenna is expressed as
	
	\begin{equation}
		\begin{aligned}
			\mathbf{a}\left(\mathbf{s}_m\right) = \left[e^{j\frac{2\pi}{\lambda}\rho^1_{\text{B}}\left(\mathbf{s}_m\right)}, e^{j\frac{2\pi}{\lambda}\rho^2_{\text{B}}\left(\mathbf{s}_m\right)}, \dots, e^{j\frac{2\pi}{\lambda}\rho^{L_{\text{B,S}}}_{\text{B}}\left(\mathbf{s}_m\right)} \right]^{T} ,
		\end{aligned}
	\end{equation}
	where $\rho^e_{\text{B}}\left(\mathbf{s}_{m}\right) = x_m\cos \theta^e_{\text{B}} \sin \varphi^e_{\text{B}}+y_m\sin \theta^e_{\text{B}}$ is the signal propagation difference quantified by the geometrical delay between $m\text{-}$th FPA
	$\mathbf{s}_m$ and $O_b$ along the $e$-th propagation path. Thus, the field response matrix (FRM) across all the FPAs at the BS is characterized by $\mathbf{A}=\begin{bmatrix}
		\mathbf{a}\left(\mathbf{s}_1\right),\mathbf{a}\left(\mathbf{s}_2\right),...,\mathbf{a}\left(\mathbf{s}_M\right)
	\end{bmatrix} \in \mathbb{C}^{L_{\text{B,S}} \times M}$. It is noted that $\mathbf{A}$ is a constant matrix given that the BS antennas' positions are fixed. 
	
	Let  $\theta^f_{\text{S,in}}$and $\varphi^f_{\text{S,in}}$ denote the elevation
	and azimuth angle-of-arrival (AoA)
	of the $f$-th BS-ME-STARS path, respectively. The receive FRV of the $n$-th ME at the STARS can be given as follows:
	\begin{equation}
		\begin{aligned}
			\mathbf{g}_{\text{S,in}}\left(\mathbf{q}_n\right) \triangleq \left[e^{j\frac{2\pi}{\lambda}\rho^1_{\text{S,in}}\left(\mathbf{q}_{n}\right)}, e^{j\frac{2\pi}{\lambda}\rho^2_{\text{S,in}}\left(\mathbf{q}_{n}\right)}, \dots, e^{j\frac{2\pi}{\lambda}\rho^{L_{\text{B,S}}}_{\text{S,in}}\left(\mathbf{q}_{n}\right)} \right]^{T} ,
		\end{aligned}
	\end{equation}
	where $\rho^f_{\text{S,in}}\left(\mathbf{q}_{n}\right) = x_n\cos \theta^f_{\text{S,in}} \sin \varphi^f_{\text{S,in}}+y_n\sin \theta^f_{\text{S,in}}$ is the signal propagation difference between the $n\text{-}$th ME
	$\mathbf{q}_n$ and the $O_s$ in the local coordinate along the $f$-th propagation path. Then, the receive FRM at the ME-STARS is given by $\mathbf{G}_{\text{S,in}}\left(\mathbf{Q}\right)=\begin{bmatrix}
		\mathbf{g}^r\left(\mathbf{q}_1\right),\mathbf{g}^r\left(\mathbf{q}_2\right),\ldots,\mathbf{g}^r\left(\mathbf{q}_n\right) 
	\end{bmatrix}\in \mathbb{C}^{L_{\text{B,S}}\times N}$. 
	Moreover, let $\boldsymbol{\Sigma}_{\text{B,S}} \in \mathbb{C}^{L_{\text{B,S}}\times L_{\text{B,S}}}$ denote the path response matrix from BS reference point $O_b$ to ME-STARS reference point $O_s$, where $\boldsymbol{\Sigma}_{\text{B,S}}\begin{bmatrix}
		e,f
	\end{bmatrix}$ characterizes the channel response between the $e$-th transmit and the $f$-th receive paths.
	Then, the channel matrix between the BS and the ME-STARS can be written as
	\begin{equation}
		\mathbf{F}\left(\mathbf{Q}\right)=\mathbf{G}_{\text{S,in}}\left(\mathbf{Q}\right)^{H} \boldsymbol{\Sigma}_{\text{B,S}} \mathbf{A} \in \mathbb{C}^{N \times M}.
	\end{equation}
	
	For the channel from  ME-STARS to user $k$, let $\theta^g_{\text{S},k}$ and $\varphi^g_{\text{S},k}$ denote the elevation and azimuth AoD
	of the $g$-th transmit path, respectively.
	By defining $\rho^g_{\text{S},k}\left(\mathbf{q}_{n}\right)= x_{n} \cos\theta^g_{\text{S},k} \sin\varphi^g_{\text{S},k} + y_{n} \sin\theta^g_{\text{S},k}$, the transmit FRV of the $n$-th MEs for the user $k$ can be given as follows:
	\begin{equation}
		\begin{aligned}
			\mathbf{g}_k\left(\mathbf{q}_n\right) \triangleq \left[e^{j\frac{2\pi}{\lambda}\rho^1_{\text{S},k}\left(\mathbf{q}_{n}\right)}, e^{j\frac{2\pi}{\lambda}\rho^2_{\text{S},k}\left(\mathbf{q}_{n}\right)}, \dots, e^{j\frac{2\pi}{\lambda}\rho^{L_{\text{S},k}}_{\text{S},k}\left(\mathbf{q}_{n}\right)} \right]^{T}.
		\end{aligned}
	\end{equation}
	Accordingly, the transmit FRM at the ME-STARS  for the user $k$ can be derived as $\mathbf{G}_k\left(\mathbf{Q}\right)=\begin{bmatrix}
		\mathbf{g}_k\left(\mathbf{q}_1\right),\mathbf{g}_k\left(\mathbf{q}_2\right),\ldots,\mathbf{g}_k\left(\mathbf{q}_N\right) 
	\end{bmatrix}\in \mathbb{C}^{L_{\text{S},k}\times N}$.  Furthermore, the path response matrix $\boldsymbol{\Sigma}_{\text{S},k} \in L_{\text{S},k} \times L_{\text{S},k}$  captures the channel characteristics from ME-STARS reference point $O_s$ to the user $k$ $\mathbf{u}_k$. Subsequently, the channel vector for the link from the ME-STARS to the user $k$ can be written as
	\begin{equation}
		\mathbf{h}_{k}\left(\mathbf{Q}\right) = \mathbf{1}_{L_{\text{S},k}}\boldsymbol{\Sigma}_{\text{S},k}\mathbf{G}_k\left(\mathbf{Q}\right)\in \mathbb{C}^{1 \times N}, \forall k\in\mathcal{K}.
	\end{equation}
	Let $\mathbf{w}_k\in\mathbb{C}^{M \times 1}$ and $x_k$  denote the BS beamforming vector and the signal intended for user $k$ with $\mathbb{E}\left\{|x_k|^2\right\}=1$, respectively. Then, the signal received at user $k$ is given by
	\begin{equation}
		\begin{aligned}
			y_{k}={\left( \mathbf{h}_{k}\left(\mathbf{Q}\right)\mathbf \Theta_s\mathbf{F}\left(\mathbf{Q}\right)\right)\sum_{k\in\mathcal{K}}\mathbf{w}_kx_k}+n_{k}, \forall k\in\mathcal{K},
		\end{aligned}
	\end{equation}
	where $\text{s}\in\{\text{t},\text{r}\}$ indicates one of the half-spaces of the ME-STARS and $s=t$ if $k\in\mathcal{T}$, otherwise $s=r$. $n_k\sim\mathcal{CN} \left(0, \sigma^2\right)$ is the additive white Gaussian noise (AWGN) at user $k$ with noise power $\sigma^2$. To better characterize the performance potential of the considered ME-STARS-assisted NOMA system, perfect CSI is assumed to be available at the BS\footnote{Given that the perfect CSI acquisition is challenging in practice~\cite{CSI-reviewer,CSI-reviewer2,CSI-reviewer3}, the impact of imperfect CSI on the performance of ME-STARS-assisted NOMA communication system deserves further investigation in future work.}.
	
%

Let $\Omega_k\in\mathcal{K}$ denote the decoding order for user $k$. For example, $\Omega_k=k$ represents that user $k$'s signal is the $k$-th to be decoded. The NOMA system fundamentally relies on the SIC technique to enable efficient multi-user interference cancellation. Specifically, users with better channel conditions sequentially decode and subtract the signals intended for users with weaker channel conditions before decoding their own signals. Therefore, $\Omega_k<\Omega_j$ is set when user $j$ has a stronger channel power gain than user $k$. Thus, the decoding order is important for the application of SIC. A straightforward way to obtain the optimal decoding order is to exhaustively search over all possible decoding sequences and solve the corresponding optimization problem for each candidate order. Then, the order yielding the best objective value is selected. However, the number of possible decoding orders increases factorially with the number of users, i.e., $\mathcal{O}(K!)$, which leads to prohibitive computational complexity. To avoid such prohibitive complexity, we adopt a distance-based decoding order design method. Specifically, the channel power gains of users are influenced by the path response matrix $\boldsymbol{\Sigma}_k$, $\forall k \in \mathcal{K}$, which is predominantly determined by the distance between user $k$ and the ME-STARS reference point $\mathbf{q}_c$. Moreover, since the signal propagation to all users involves the BS-STARS channel part, the differences among users' channel power gains mainly come from the ME-STARS-user channels. Thus, users farther away from the ME-STARS are assigned lower decoding orders, which can be expressed as follows:
	\begin{equation}
		\Omega_j<\Omega_k, \text{if} \quad \| \mathbf{u}_k-\mathbf{q}_c\| \le \| \mathbf{u}_j-\mathbf{q}_c\|.
	\end{equation}
	With the obtained decoding order, user $k$'s achievable signal-to-interference-plus-noise ratio (SINR) is given by
	\begin{equation}
		{\rm SINR}_{k\rightarrow k}=\frac{{\left|\left( \mathbf{h}_{k}\left(\mathbf{Q}\right)\mathbf \Theta_s\mathbf{F}\left(\mathbf{Q}\right)\right)\mathbf{w}_k\right|^2}}{{\sum\limits_{\Omega_i>\Omega_k}\left|\left( \mathbf{h}_{k}\left(\mathbf{Q}\right)\mathbf \Theta_s\mathbf{F}\left(\mathbf{Q}\right)\right)\mathbf{w}_i\right|^2}+\sigma^2}.
		\label{T}
	\end{equation}
	The corresponding achievable rate at user $k$ to decode its own signal is $R_{k\rightarrow k}=\log_2\left(1+{\rm SINR}_{k\rightarrow k}\right)$. 
	Furthermore, for users $k,j\in\mathcal{K}$ with $\Omega_k<\Omega_j$, user $j$'s  SINR for decoding user $k$’s signal is 
	\begin{equation}
		{\rm SINR}_{j\rightarrow k}=\frac{{\left|\left( \mathbf{h}_{j}\left(\mathbf{Q}\right)\mathbf \Theta_s\mathbf{F}\left(\mathbf{Q}\right)\right)\mathbf{w}_k\right|^2}}{{\sum\limits_{\Omega_i>\Omega_k}\left|\left( \mathbf{h}_{j}\left(\mathbf{Q}\right)\mathbf \Theta_s\mathbf{F}\left(\mathbf{Q}\right)\right)\mathbf{w}_i\right|^2}+\sigma^2}.
		\label{T2}
	\end{equation}
	The corresponding achievable rate at user $j$ to decode user $k$'s signal is $R_{j\rightarrow k}=\log_2\left(1+{\rm SINR}_{j\rightarrow k}\right)$. Moreover, to guarantee the successful SIC at users, we define $R_k=\min\left\{R_{k\rightarrow k},R_{j\rightarrow k}\right\},  \;{\rm if}\;\Omega_k<\Omega_j, \forall k,j\in\mathcal{K}$ as the achievable rate of user $k$ that is limited to the decoding rates at user $k$ and users with higher decoding orders for user $k$'s message~\cite{NGMA}.
	
	Furthermore, to guarantee the rate fairness among users, the preferential resources should be allocated to users with weaker channel conditions~\cite{NOMA5}, which yields
	\begin{align}
		\label{fair}
		&{\left|\left( \mathbf{h}_{i}\left(\mathbf{Q}\right)\mathbf \Theta_s\mathbf{F}\left(\mathbf{Q}\right) \right)\mathbf{w}_k\right|^2}\ge{\left|\left( \mathbf{h}_{i}\left(\mathbf{Q}\right)\mathbf \Theta_s\mathbf{F}\left(\mathbf{Q}\right)\right)\mathbf{w}_j\right|^2},\notag\\
		& \;{\rm if}\;\Omega_k<\Omega_j,\forall i,k,j\in\mathcal{K}, \forall s\in\{t,r\}.
	\end{align}
	The constraint in \eqref{fair} ensures that an adequate decoding rate is attainable for users with lower decoding orders. 
	
	

	\subsection{Problem Formulation }
	This paper aims to maximize the sum achievable rate for all $K$ users through joint designing of the BS beamforming, the ME-STARS beamforming, and the MEs positions. The resulting optimization problem can be formulated as
	\begin{subequations}\label{P1}
		\begin{equation}
			\mathop{\rm{max}}\limits_{\left\{\mathbf{w}_k,\mathbf{ \Theta}_{s}, \mathbf{Q}\right\}}\quad \sum_{k=1}^K R_k,
		\end{equation}
		\begin{equation}
			\label{P1_C9}
			{\rm s.t.} \;
			R_k=\min\left\{R_{k\rightarrow k},R_{j\rightarrow k}\right\},  \;{\rm if}\;\Omega_k<\Omega_j,\forall k,j\in\mathcal{K},
		\end{equation}
		\begin{equation}
			\label{P1_C1}\sum_{k\in\mathcal{K}}\left\|\mathbf{w}_k\right\|^2\le P_{\text{B}},
		\end{equation}
		\begin{align}
			\label{P1_C2}
			&\left|\left( \mathbf{h}_{j}\left(\mathbf{Q}\right)\mathbf \Theta_s \mathbf{F}\left(\mathbf{Q}\right) \right) \mathbf{w}_k\right|^2\ge{\left|\left( \mathbf{h}_{j}\left(\mathbf{Q}\right)\mathbf \Theta_s\mathbf{F}\left(\mathbf{Q}\right)\right)\mathbf{w}_i\right|^2},\notag\\
			&\;{\rm if}\;\Omega_k<\Omega_i,\forall i,k,j\in\mathcal{K}, \forall s\in\{t,r\},
		\end{align}
		\begin{equation}
			\label{P1_C5}
			\theta_n^t, \theta_n^r\in[0,2\pi), \forall n\in\mathcal{N},
		\end{equation}
		\begin{equation}
			\label{P1_C6}
			\beta_n^t,\beta_n^r\in\left[0,1\right], \beta_n^t+\beta_n^r=1,\forall n\in\mathcal{N},\end{equation}
		\begin{equation}
			\label{P1_C7}
			\mathbf{q}_{n} \in \mathcal{C}_{\rm ME},\forall n\in\mathcal{N},
		\end{equation}
		\begin{equation}
			\label{P1_C8}
			\|\mathbf{q}_{n}-\mathbf{q}_{l}\|_2 \ge D_{\min},\forall n,l\in\mathcal{N},	n\neq l,
		\end{equation}
	\end{subequations}
	where $P_{\text{B}}$ represents the maximum transmit power allowed at the BS and $D_{\text{min}}$ is the minimum distance. Constraint \eqref{P1_C9} ensures successful user signal decoding through SIC implementation. Constraint \eqref{P1_C1}  and constraint \eqref{P1_C2} are the total transmit power constraint and rate fairness constraint. The phase shift and energy conservation constraints for each ME are given in \eqref{P1_C5} and \eqref{P1_C6}, respectively.
	Constraint \eqref{P1_C7} guarantees
	that all MEs  move within the feasible region $\mathcal{C}_{\rm ME}$. Constraint
	\eqref{P1_C8} ensures any adjacent MEs keep the minimum distance. 
	
	We further convert problem \eqref{P1} into a more tractable formulation as follows:
	\begin{subequations}
		\label{P15}
		\begin{equation}
			\mathop{\rm{max}}\limits_{\left\{\mathbf{w}_k,\mathbf{ \Theta}_{s}, \mathbf{Q},R_k^{\text{min}}\right\}}\quad \sum_{k=1}^K R_k^{\text{min}},   
		\end{equation}  
		\begin{equation}
			\label{P15_C1}{\rm s.t.} \;
			R_k^{\text{min}}\le R_{j\rightarrow k},  \;{\rm if}\;\Omega_k \leq\Omega_j, \forall k,j\in\mathcal{K},	
		\end{equation}
		\begin{equation}\label{P15_C2}\eqref{P1_C1}-\eqref{P1_C8}.
		\end{equation}
	\end{subequations}
	Problem \eqref{P15} is a non-convex problem with  highly-coupled variables, which makes it challenging to resolve. Note that, the newly introduced MEs positions need to be jointly optimized with the BS and STARS beamforming, which brings in additional complexity. In the following sections, we will develop an AO algorithm for effectively solving problem \eqref{P15}.
	\section{Proposed Solution}
	In this section, we invoke the AO framework to decouple the problem \eqref{P15} into three subproblems, i.e., BS beamforming, ME-STARS beamforming, and MEs positions optimization. The three subproblems are alternatively solved in an iterative procedure until the sum achievable rate stabilizes according to predefined criteria. We then delineate the general algorithmic procedure and conduct a theoretical analysis to investigate its convergence performance and computational complexity.
	\subsection{BS Beamforming Optimization}
	\vspace{-5pt}
	
	Given MEs positions $\left\{\mathbf{Q}\right\}$ and ME-STARS beamforming $\left\{\mathbf{\Theta_s}\right\}$, let $\mathbf{\overline{h}}_k= \mathbf{h}_{k}\left(\mathbf{Q}\right)\mathbf \Theta_s\mathbf{F}\left(\mathbf{Q}\right)$ denote the combined channel of the BS-STARS-user link for user $k$. To convert the initial fractional-form problem into a more tractable one, a
	slack variable set $\left\{M_{jk},N_{jk}| \forall j,k\in \mathcal{K}\right\}$ is introduced, where $M_{jk}$ and $N_{jk}$ are defined as
	\begin{equation}
		\frac{1}{M_{jk}}=\left|\mathbf{\overline{h}}_j\mathbf{w}_k\right|^2 ,\forall j,k\in \mathcal{K},
		\label{A.1}
	\end{equation}
	\begin{equation}
		N_{jk}={\sum\limits_{\Omega_i>\Omega_k}\left|\mathbf{\overline{h}}_j\mathbf{w}_i\right|^2}+\sigma^2,\forall j,k\in \mathcal{K}.
		\label{A.2}
	\end{equation}
	Substituting \eqref{A.1} and \eqref{A.2} into \eqref{T2}, the achievable rate
	can be reformulated as
	\begin{align}
		R_{j\rightarrow k}=&\log_2\left(1+\frac{1}{M_{jk}N_{jk}}\right),\notag \\
		&\;{\rm if}\;\Omega_k \leq\Omega_j, \forall j,k\in \mathcal{K}.
		\label{A.3}
	\end{align}
	We further define $\mathbf{H}_k=\mathbf{\overline{h}}_k^H\mathbf{\overline{h}}_k$ and $\mathbf{W}_k=\mathbf{w}_k\mathbf{w}_k^H$, where $\mathbf{W}_k \succeq 0$ and $\text{rank}\left(\mathbf{W}_k\right)=1$. It follows that $\left|\mathbf{\overline{h}}_k\mathbf{w}_k\right|^2= \rm{Tr}\left(\mathbf{W}_k\mathbf{H}_k\right)$. Thus, the BS beamforming problem can be equivalently rewritten as
	\begin{subequations}\label{P3}
		\begin{equation}
			\mathop{\rm{max}}\limits_{\left\{\mathbf{W}_k,M_{jk},N_{jk},R_k^{\text{min}}\right\}}\quad \sum_{k=1}^K R_k^{\text{min}},
		\end{equation}
		\begin{align}
			\label{P3_C1} {\rm s.t.} \;\log_2&\left(1+\frac{1}{M_{jk}N_{jk}}\right) \geq R_k^{\text{min}},\notag\\ & \;{\rm if}\;\Omega_k \leq\Omega_j, \forall j,k \in \mathcal{K},
		\end{align}
		\begin{equation}
			\label{P3_C2} \frac{1}{M_{jk}} \leq \text{Tr}\left(\mathbf{W}_k\mathbf{H}_j\right),\forall j,k \in \mathcal{K}, 
		\end{equation}
		\begin{equation}
			\label{P3_C3} N_{jk}\geq {\sum\limits_{\Omega_i>\Omega_k}}\text{Tr}\left(\mathbf{W}_i\mathbf{H}_j\right)+\sigma^2,\forall j,k \in \mathcal{K},
		\end{equation}
		\begin{equation}
			\label{P3_C4}
			\sum_{k\in\mathcal{K}}\rm{Tr}\left(\mathbf{W}_k\right)\le P_{\text{B}},
		\end{equation}
		\begin{equation}
			\label{P3_C5} \mathbf{W}_k \succeq 0,
		\end{equation}
		\begin{equation}
			\label{P3_C6} \text{rank}\left(\mathbf{W}_k\right)=1,
		\end{equation}
		\begin{equation}
			\label{P3_C7}
			\text{Tr}\left(\mathbf{H}_j\mathbf{W}_k \right) \ge \text{Tr}\left(\mathbf{H}_j\mathbf{W}_i\right),\;{\rm if}\;\Omega_k<\Omega_i,\forall i,k,j\in\mathcal{K}.
		\end{equation}
	\end{subequations}
	It is noted that problem \eqref{P3} still non-convex due to constraints \eqref{P3_C1} and \eqref{P3_C6}. To handle the constraint \eqref{P3_C1}, a convex relaxation is derived by applying the first-order Taylor expansion to approximate the left-hand term of \eqref{P3_C1} as follows:
	\begin{align}
		\label{taylor}
		&\log _2\left(1+\frac{1}{M_{jk} N_{jk}}\right)  \geqslant \log _2\left(1+\frac{1}{M_{jk}^{\left(\tau_1\right)} N_{jk}^{\left(\tau_1\right)}}\right)\nonumber\\
		&-\frac{\log _2 \text{e}\left(M_{jk}-M_{jk}^{\left(\tau_1\right)}\right)}{M_{jk}^{\left(\tau_1\right)}\left(1+M_{ jk}^{\left(\tau_1\right)} N_{jk}^{\left(\tau_1\right)}\right)}-\frac{\log _2 \text{e}\left(N_{jk}-N_{ jk}^{\left(\tau_1\right)}\right)}{N_{jk}^{\left(\tau_1\right)}\left(1+M_{jk}^{\left(\tau_1\right)} N_{jk}^{\left(\tau_1\right)}\right)} \nonumber\\
		& =\widetilde{R}_{j \rightarrow k},  \;{\rm if}\;\Omega_k \leq\Omega_j, \forall j,k \in \mathcal{K},
	\end{align}
where the points $M^{\left(\tau_1\right)}_{jk}$ and $N^{\left(\tau_1\right)}_{jk}$ are given local points of $M_{jk}$ and $N_{jk}$ in the $\tau_1$-th iteration, respectively. Then, the optimization subproblem \eqref{P3} can be rewritten as
	\begin{subequations}\label{P4}
		\begin{equation}
			\mathop{\rm{max}}\limits_{\left\{\mathbf{W}_k,M_{jk},N_{jk},R_k^{\text{min}}\right\}}\quad \sum_{k=1}^K R_k^{\text{min}},
		\end{equation}
		\begin{equation}
			\label{P4_C1}{\rm s.t.} \;
			\widetilde{R}_{j \rightarrow k} \geq R_k^{\text{min}},  \;{\rm if}\;\Omega_k \leq\Omega_j,  \forall j,k \in \mathcal{K},
		\end{equation}
		\begin{equation}
			\label{P4_C2} \eqref{P3_C2}-\eqref{P3_C7}.
		\end{equation}
	\end{subequations}
	Subsequently, the rank-one constraint \eqref{P3_C6} remains the sole non-convex obstacle for solving problem \eqref{P4}.  To address this issue, we have the following theorem.
	\begin{theorem}
		The optimal solution $\mathbf{W}_k^*$ with respect to problem \eqref{P4} inherently satisfies $\text{rank}\left(\mathbf{W}_k\right)\leq1$ despite omitting the rank-one constraint \eqref{P3_C6} .
	\end{theorem}
	\begin{proof}
		The proof is omitted for brevity. For more details, we refer the readers to~\cite{proof}.
	\end{proof}
	Theorem 1 establishes that the optimal beamforming solution
	${\mathbf{W}_k^*}$ for problem \eqref{P3} can be obtained by solving convex problem \eqref{P4} without the rank-one constraint.  This tractable formulation enables efficient computation via standard solvers  such as the semi-definite program solver in CVX tool~\cite{cvx_software}. The detailed steps of the proposed algorithm are shown in \textbf{Algorithm~\ref{alg:pro1}}.
	\begin{algorithm}
		\caption{SCA-based algorithm for solving problem \eqref{P3}}
		\label{alg:pro1}
		\begin{algorithmic}[1]
			\STATE Initialize feasible points $\left\{M_{jk}^{\left(0\right)}\right\}$, $\left\{N_{jk}^{\left(0\right)}\right\}$ and $\mathbf{W}_k^{\left(0\right)}$;
			\STATE Set the iteration index $\tau_1 = 0$;
			\REPEAT
			\STATE Update $\left\{M_{jk}^{\left(\tau_1 + 1\right)}\right\}$, $\left\{N_{jk}^{\left(\tau_1 + 1\right)}\right\}$ and $\left\{\mathbf{W}_k^{\left(\tau_1 + 1\right)}\right\}$ by solving problem \eqref{P4};
			\STATE $\tau_1 \gets \tau_1 + 1$;
			\UNTIL{the fractional increment of the objective value is less than the predefined threshold $\varepsilon_1$ or $\tau_1\geq\tau_1^{\text{max}}$.}
			\STATE \textbf{Output:} $\mathbf{W}_k^*$
		\end{algorithmic}
	\end{algorithm}
	
	\subsection{ME-STARS Beamforming Optimization}
	Given MEs positions $\{\mathbf{Q}\}$ and transmit beamforming vectors $\{\mathbf{w}_k\}$, we adopt the same slack variable set $\{M_{jk},N_{jk}| \forall j,k\in \mathcal{K}\}$ as  in the previous subsection. By denoting $\mathbf{u}_{s}=\begin{bmatrix}
		\sqrt{\beta_1^s}e^{j\theta_{1}^s},\sqrt{\beta_2^s}e^{j\theta_{2}^s},\ldots,\sqrt{\beta_N^s}e^{j\theta_{N}^s}
	\end{bmatrix} $
	and $\mathbf{Z}_j\left(\mathbf{Q}\right)=\text{diag}\left(\mathbf{h}_j\left(\mathbf{Q}\right)\right)\mathbf{F}\left(\mathbf{Q}\right)$, we have  $\left|\left( \mathbf{h}_{j}\left(\mathbf{Q}\right)\mathbf \Theta_s\mathbf{F}\left(\mathbf{Q}\right) \right)\mathbf{w}_k\right|^2={\left|\left(\mathbf{u}_{s}^H\mathbf{Z}_j\left(\mathbf{Q}\right)\right)\mathbf{w}_k\right|^2}$. We further define $\mathbf{U}_s=\mathbf{u}_s^H\mathbf{u}_s$ and $\mathbf{V}_{j,k}=\mathbf{Z}_j\left(\mathbf{Q}\right)\mathbf{w}_k\mathbf{w}_k^H\mathbf{Z}_j\left(\mathbf{Q}\right)^H$, where  $\mathbf{U}_s \succeq 0$ , $\text{rank}\left(\mathbf{U}_s\right)=1$ and $\text{diag}\left(\mathbf{U}_s\right)=\left[\beta_1^s,\beta_2^s,...,\beta_N^s\right]$.
	The ME-STARS beamforming is given by
	
	\begin{subequations}\label{P5}
		\begin{equation}
			\mathop{\rm{max}}\limits_{\left\{\mathbf{U}_s,M_{jk},N_{jk},R_k^{\text{min}}\right\}}\quad \sum_{k=1}^K R_k^{\text{min}},
		\end{equation}
		\begin{align}
			\label{P5_C1} {\rm s.t.} \;
			\log_2&\left(1+\frac{1}{M_{jk}N_{jk}}\right) \geq R_k^{\text{min}},\notag\\
			&\;{\rm if}\;\Omega_k \leq\Omega_j, \forall j,k\in \mathcal{K},
		\end{align}
		\begin{equation}
			\label{P5_C2}
			\frac{1}{M_{jk}}\leq{\text{Tr}\left(\mathbf{U}_s\mathbf{V}_{j,k}\right)},\forall j,k\in \mathcal{K},
		\end{equation}
		\begin{equation}
			\label{P5_C3} N_{jk}\geq \sum\limits_{\Omega_i>\Omega_k}{\text{Tr}\left(\mathbf{U}_s\mathbf{V}_{j,i}\right)}+\sigma^2,\forall j,k\in \mathcal{K},
		\end{equation}
		\begin{align}
			\label{P5_C4} {\text{Tr}\left(\mathbf{U}_s\mathbf{V}_{j,k}\right)} &\ge{\text{Tr}\left(\mathbf{U}_s\mathbf{V}_{j,i}\right)},\;{\rm if}\;\Omega_k<\Omega_i,\notag\\
			&
			\forall i,k,j\in\mathcal{K}, \forall s\in\{t,r\},
		\end{align}
		\begin{equation}
			\label{P5_C5} \mathbf{U}_s \succeq 0, 
		\end{equation}
		\begin{equation}
			\label{P5_C6}  \left(\mathbf{U}_s\right)\left[m,m\right]=\beta_m^s,
		\end{equation}
		\begin{equation}
			\label{P5_C7} \text{rank}\left(\mathbf{U}_s\right)=1, 
		\end{equation}
		\begin{equation}
			\label{P5_C8}         \eqref{P1_C5}, \eqref{P1_C6}.
		\end{equation}
	\end{subequations}
	It is obvious that problem \eqref{P5} is non-convex due to the constraints \eqref{P5_C1} and \eqref{P5_C7}. Constraint \eqref{P5_C1} can be approximately transformed into a convex constraint by combining 
	\eqref{taylor} as shown in the last subsection.  As for the rank-one constraint \eqref{P5_C7}, we can transform it into an equality form due to the special property of the rank-one matrix, which is given by
	\begin{equation}
		\label{equ}
		\text{Tr}\left(\mathbf{U}_s\right)-||\mathbf{U}_s||_2=0, \forall s \in \{r,t\},
	\end{equation}
	where $||\mathbf{U}_s||_2=\xi_{\text{max}}\left(\mathbf{U}_s\right)$ denotes the largest eigenvalue of $\mathbf{U}_s$. It is noted that, for any matrix $\mathbf{B}$ that is not a rank-one matrix, $\text{Tr}\left(\mathbf{B}\right)-||\mathbf{B}||_2>0$ is always satisfied. Thus, to address problem \eqref{P5}, a penalty-based method is employed, which involves converting the equality constraint \eqref{equ} into a penalty term integrated with the objective function. Denoting $\kappa>0$ as the penalty weight, the problem \eqref{P5} can be reformulated as
	\begin{subequations}\label{P6}
		\begin{align}
			&\mathop{\rm{max}}\limits_{\left\{\mathbf{U}_s,M_{jk},N_{jk},R_k^{\text{min}}\right\}}\quad \sum_{k=1}^K R_k^{\text{min}}\notag\\
			&-\kappa\sum_{s\in \{ r,t\}}\text{Tr}\left(\mathbf{U}_s\right)-||\mathbf{U}_s||_2,
		\end{align}
		\begin{equation}
			\label{P6_C1}{\rm s.t.} \;
			\widetilde{R}_{j\rightarrow k}\ge R_k^{\text{min}},   \;{\rm if}\;\Omega_k \leq\Omega_j, \forall j,k \in \mathcal{K},
		\end{equation}
		\begin{equation}
			\label{P6_C2} 
			\eqref{P1_C5},\eqref{P1_C6}
			, \eqref{P5_C2}-\eqref{P5_C6}.	
		\end{equation}
	\end{subequations}
	Subsequently, to tackle the non-convex penalty term, the first-order Taylor expansion is employed to approximate $||\mathbf{U}_s||_2$ to its linear upper bound as
	\begin{equation}
		\begin{aligned}
			\text{Tr}\left(\mathbf{U}_s\right)-||\mathbf{U}_s||_2\le \text{Tr}\left(\mathbf{U}_s\right)-||\mathbf{U}_s||^{\left(\tau_2\right)}_{e}=\varepsilon_s^{\left(\tau_2\right)},
		\end{aligned}
	\end{equation}
	where $||\mathbf{U}_s||^{\left(\tau_2\right)}_{e}=||\mathbf{U}_s^{\left(\tau_2\right)}||_2+\text{Tr}\left(\mathbf{u}_{\text{max}}^{\left(\tau_2\right)}\mathbf{u}_{\text{max}}^{\left(\tau_2\right)H}\left(\mathbf{U}_s-\mathbf{U}_s^{\left(\tau_2\right)}\right)\right)$ with $\mathbf{u}_{\text{max}}^{\left(\tau_2\right)}$ denoting the eigenvector corresponding to the largest eigenvalue of $\mathbf{U}_s^{\left(\tau_2\right)}$. Accordingly, a equivalent form of problem \eqref{P6} can be given as 
	\begin{subequations}\label{P7}
		\begin{equation}
			\mathop{\rm{max}}\limits_{\left\{\mathbf{U}_s,M_{jk},N_{jk},R_k^{\text{min}}\right\}}\quad \sum_{k=1}^K R_k^{\text{min}}-\kappa\sum_{s\in \{ r,t\}}\varepsilon_s^{\left(\tau_2\right)},
		\end{equation}
		\begin{equation}
			\label{P7_C1}{\rm s.t.} \;
			\eqref{P1_C5},\eqref{P1_C6}, \eqref{P5_C2}-\eqref{P5_C6},\eqref{P6_C1}.
		\end{equation}
	\end{subequations}
	
	It is noted that, when the penalty weight is sufficiently large, the penalty term could ensure that the obtained solution satisfies the rank-one condition. However, if the penalty weight $\kappa$ is initialized with an sufficiently large value, the objective function will become predominantly governed by the penalty term. This dominance might limit the exploration of the solution space, thereby leading to insufficient consideration of the sum achievable rate. Thus, a two-layer algorithm is proposed, where $\left\{\mathbf{U}_t, \mathbf{U}_r\right\}$ are updated by solving $\eqref{P7}$ in the inner layer. As for the penalty factor $\kappa$, it is set to be a relatively small value and iteratively increased in the outer layer with $\kappa=l_{\kappa}\kappa, l_{\kappa}>1$, until the violation of the equality constraint in \eqref{equ} gets smaller than the predetermined tolerance threshold $\epsilon$ i.e., $\text{max}\left\{\text{Tr}\left(\mathbf{U}_t\right)-||\mathbf{U}_t||_2,\text{Tr}\left(\mathbf{U}_r\right)-||\mathbf{U}_r||_2\right\}\le \epsilon$.
	
	It is obvious that the problem \eqref{P7} is a standard convex problem after transformation, which can be efficiently solved by convex problem solvers like the CVX tool~\cite{cvx_software}. Specifically, the proposed algorithm for solving ME-STARS beamforming optimization problem is shown in \textbf{Algorithm~\ref{alg:pro2}}.
	\begin{algorithm}
		\caption{Penalty-based SCA algorithm for solving problem \eqref{P5}}
		\label{alg:pro2}
		\begin{algorithmic}[1]
			\STATE Initialize feasible points $\left\{M_{jk}^{\left(0\right)}\right\}$, $\left\{N_{jk}^{\left(0\right)}\right\}$, $ \left\{\mathbf{U}_{r}^{\left(0\right)},\mathbf{U}_{t}^{\left(0\right)}\right\}$ and the penalty factor $\kappa$ ;
			\REPEAT
			\STATE Set the iteration index $\tau_2 = 0$ for inner loop;
			\REPEAT
			\STATE With the given $\left\{M_{jk}^{\left(\tau_2 \right)}\right\}$, $\left\{N_{jk}^{\left(\tau_2 \right)}\right\}$ and $ \left\{\mathbf{U}_{r}^{\left(\tau_2\right)},\mathbf{U}_{t}^{\left(\tau_2\right)}\right\}$, solve the relaxed problem \eqref{P7};
			\STATE Update $\left\{M_{jk}^{\left(\tau_2 +1\right)}\right\}$, $\left\{N_{jk}^{\left(\tau_2+1\right)}\right\}$ and $ \left\{\mathbf{U}_{r}^{\left(\tau_2+1\right)},\mathbf{U}_{t}^{\left(\tau_2+1\right)}\right\}$;
			\STATE $\tau_2 \gets \tau_2 + 1$
			\UNTIL{the fractional increment of the objective value is less than the predefined threshold $\varepsilon_1$ or $\tau_2 \geq \tau_2^{\text{max}}$. }
			\STATE Update $\mathbf{U}_{r}^{\left(0\right)}=\mathbf{U}_{r}^{\left(\tau_2\right)}, \mathbf{U}_{t}^{\left(0\right)}=\mathbf{U}_{t}^{\left(\tau_2\right)}$ and $\kappa=l_{\kappa}\kappa$;
			\UNTIL{$\text{max}\left \{\text{Tr}\left(\mathbf{U}_t\right)-||\mathbf{U}_t||_2,\text{Tr}\left(\mathbf{U}_r\right)-||\mathbf{U}_r||_2\right \}\le \epsilon$.}
			\STATE \textbf{Output:} $\mathbf{U}_t^*$, $\mathbf{U}_r^*$
		\end{algorithmic}
	\end{algorithm}
	
	\subsection{MEs Positions Optimization}
	With the given BS beamforming $\left\{\mathbf{w}_k\right\}$ and ME-STARS beamforming coefficient $\left\{\mathbf{\Theta}_{s}\right\}$, we define $\mathbf{E}\left(\mathbf{q}_n\right)=\mathbf{g}_j\left(\mathbf{q}_{n}\right)\mathbf{g}_{\text{S,in}}\left(\mathbf{q}_{n}\right)^{H} \in\mathbb{C}^{L_{\text{S},k}\times L_{\text{B,S}}}$, where $\mathbf{E}\left(\mathbf{q}_n\right)\left[i,j\right]=e^{j\frac{2\pi}{\lambda}\left(\rho^i_{\text{S},k}\left(\mathbf{q}_{n}\right)-\rho^j_{\text{S,in}}\left(\mathbf{q}_{n}\right)\right)}$. In addition, we let $\mathbf{\omega}_{n,j}=\sqrt{\beta_n^s}e^{j\theta_{n}^s}\mathbf{1}_{L_{\text{S},k}}\boldsymbol{\Sigma}_j$ and $\mathbf{\nu}_k=\boldsymbol{\Sigma}_{\text{B,S}} \mathbf{A}\mathbf{w}_k$ to further simplify the calculation. Thus, ${\left|\left(\mathbf{u}_{s}^H\mathbf{Z}_j\left(\mathbf{Q}\right)\right)\mathbf{w}_k\right|^2}$ can be rewritten as follows:
	\begin{equation}
		\label{trans}
		{\left|\left(\mathbf{u}_{s}^H\mathbf{Z}_j\left(\mathbf{Q}\right)\right)\mathbf{w}_k\right|^2}=\left|  \mathbf{\omega}_{n,j} \mathbf{E}\left({\mathbf{q}_n}\right)\mathbf{\nu}_k+\alpha_{n,k}\right|^2,
	\end{equation}
	where $\alpha_{n,k}=\sum\limits_{\overline{n}\neq n}^N\mathbf{\omega}_{\overline{n},j}\mathbf{E}\left(\mathbf{q}_{\overline{n}}\right)\mathbf{\nu}_k$. Note that only  $\mathbf{E}\left(\mathbf{q}_n\right)$ is determined by the position of the $n$-th ME in \eqref{trans}. Accordingly, the achievable rate at user $j$ to decode user $k$'s signal can be represented as
	\begin{equation}
		R_{j \rightarrow k}=\log_2{\left(1+\frac{|  \mathbf{\omega}_{n,j} \mathbf{E}\left({\mathbf{q}_n}\right)\mathbf{\nu}_k+\alpha_{n,k}|^2}{\sum\limits_{\Omega_i>\Omega_k}|   \mathbf{\omega}_{n,j} \mathbf{E}\left({\mathbf{q}_n}\right)\mathbf{\nu}_i+\alpha_{n,i}|^2+\sigma^2}\right)}.
	\end{equation}
	Then, the MEs positions optimization problem can be formulated as
	\begin{subequations}\label{P8}
		\begin{equation}
			\mathop{\rm{max}}\limits_{ \left\{ \mathbf{Q},R_k^{\text{min}} \right\}}\quad \sum_{k=1}^K R_k^{\text{min}},
		\end{equation}
		\begin{equation}
			\label{P8_C1}{\rm s.t.}\;
			\eqref{P1_C2},\eqref{P1_C7},\eqref{P1_C8},\eqref{P15_C1}.
		\end{equation}
	\end{subequations}
	The problem \eqref{P8} is inherently non-convex w.r.t  $\mathbf{Q}$, which makes it hard to find a global optimal solution. Thus, we propose the GDA to find the maximum point in the solution space efficiently. Before detailing the GDA, we first tackle the constraints of the problem \eqref{P8} due to the limitation that GDA is not available for solving constrained problems. For the movement space constraint \eqref{P1_C7},  we introduce the Sigmoid function and the auxiliary variable set $\left\{ \mathbf{\overline{q}}_n,1\leq n \leq N\right\}$, such that:
	\begin{align}
		\mathbf{q}_n
		&=A\text{Sigmoid}\left(\mathbf{\overline{q}}_n\right) \notag\\
		&=A\left[\frac{1}{1+e^{-\overline{x}_n}},\frac{1}{1+e^{-\overline{y}_n}}\right].
	\end{align}
	The above equation enables the variable set $\left\{ \mathbf{\overline{q}}_n,1\leq n \leq N\right\}$ defined in the real space to be mapped to the confined real space $\mathcal{C}_{\text{ME}}$. Accordingly, by applying the relationship between the $\mathbf{q}_n$ and $\overline{\mathbf{q}}_n$, an equivalent form of the minimum distance constraint \eqref{P1_C8} can be derived as
	\begin{align}
		\label{P}
		d\left(\mathbf{\overline{q}}_n, \overline{\mathbf{q}}_{l}\right)=&\frac{ D_{\text{min}}}{A}-\left\|\text{Sigmoid} \left(\overline{\mathbf{q}}_n\right)-\text{Sigmoid}\left(\overline{\mathbf{q}}_{l}\right)\right\|_2 \leq 0, \notag\\
		&\forall n,l \in N,n\neq l.
	\end{align}
	Similarly, constraints \eqref{P1_C2} and \eqref{P15_C1} can be respectively rewritten as
	\begin{subequations}\label{P9}
		\begin{equation}
			\label{P9_C2}p\left (R_k^{\text{min}}\right )=\sum_{k=1}^K\sum_{\Omega_k \leq\Omega_j}R_k^{\text{min}}-R_{j\rightarrow k} \le 0 ,
		\end{equation}
		\begin{align}
			\label{P9_C3}
			r\left(\overline{\mathbf{q}}_n\right)=&\sum_{j =1}^{K}\sum_{\Omega_k<\Omega_i}\left|  \mathbf{\omega}_{n,j} \mathbf{E}\left(\overline{\mathbf{q}}_n\right)\mathbf{\nu}_i+\alpha_{n,i}\right|^2\notag\\&-\left|  \mathbf{\omega}_{n,j} \mathbf{E}\left(\mathbf{\overline{q}}_n\right)\mathbf{\nu}_k+\alpha_{n,k}\right|^2\le 0.
		\end{align}
	\end{subequations}
	
	By adding penalty terms $\kappa\text{max}\left \{0, d\left (\mathbf{\overline{q}}_n,\mathbf{\overline{q}}_l\right )\right \}$, $\kappa\text{max}\left \{0,p\left(R_k^{\text{min}}\right) \right \}$, and $\kappa\text{max}\left \{0, r\left(\overline{\mathbf{q}}_n\right)\right \}$ into the objective function, with $\kappa>0$ representing the penalty weight, a equivalent unconstrained form of problem \eqref{P8} can be derived as follows:
	\begin{align}
		\label{P10}
		&\mathop{\rm{max}}\limits_{ \left \{\mathbf{\overline{Q}},R_k^{\text{min}}\right \}}\quad f\left (\mathbf{\overline{Q}},R_k^{\text{min}}\right)=\sum_{k=1}^K R_k^{\text{min}}-\kappa\text{max}\left \{0, d\left (\mathbf{\overline{q}_n},\mathbf{\overline{q}_l}\right )\right \}\notag\\
		&
		-\kappa \text{max}\left \{0, r\left (\overline{\mathbf{q}}\right ) \right \}-\kappa \text{max}\left \{0, p\left (R_k^{\text{min}}\right ) \right \}.
	\end{align}
	Since the objective function in \eqref{P10} is not  differentiable given that $\text{max}\left\{ a,b\right\}$ function is a step function, the  log-sum-exp (LSE) is invoked as a smooth function to approximate the $\text{max}\left \{a,b\right \}$ function. Accordingly, problem \eqref{P10} can be rewritten as follows:
	\begin{align}
		\label{P11}
		&\max_{\left\{ \mathbf{\overline{Q}}, R_k^{\text{min}} \right\}} \quad f\left( \mathbf{\overline{Q}}, R_k^{\text{min}} \right) =\sum_{k=1}^K R_k^{\text{min}} - \kappa\left(  \zeta\ln\left(1+e^{\frac{d\left(\mathbf{\overline{q}}_n, \overline{\mathbf{q}}_{l}\right)}{\zeta}}\right) \right. \notag\\
		& \left.  + \zeta \ln\left(1+e^{\frac{r\left(\overline{\mathbf{Q}}\right)}{\zeta}}\right) + \zeta \ln\left(1+e^{\frac{p\left(R_k^{\text{min}}\right)}{\zeta}}\right) \right),
	\end{align}
	where $\zeta>0$ is the smoothing factor. 
	\begin{figure*}[ht]
		\setcounter{equation}{34}
		\begin{equation}
			\label{dao1}
			\frac{\partial f\left( \mathbf{Q},R_k^{\text{min}}\right)}{\partial \mathbf{q}_n}= 
			\kappa\left(\frac{e^{p\left(R_k^{\text{min}}\right) / \zeta}}{1+e^{p\left(R_k^{\text{min}}\right) / \zeta}}\sum_{k=1}^K\sum_{\Omega_k \leq\Omega_j}\frac{\partial R_{j \rightarrow k}}{\partial \mathbf{q}_n}-\sum_{l \neq n} \frac{e^{d\left(\mathbf{q}_n, \mathbf{q}_{l}\right) / \zeta}}{1+e^{d\left(\mathbf{q}_n, 
					\mathbf{q}_{l}\right) / \zeta}}\frac{\mathbf{q}_n-\mathbf{q}_l}{\sqrt{\left(\mathbf{q}_n-\mathbf{q}_l\right)^T\left(\mathbf{q}_n-\mathbf{q}_l\right)}}-\frac{e^{r\left(\mathbf{q}_n\right) / \zeta}}{1+e^{r\left(\mathbf{q}_n\right) / \zeta}}\frac{\partial r\left(\mathbf{q}_n\right)}{\partial \mathbf{q}_n}\right),
		\end{equation}
		\setcounter{equation}{36}
		\begin{equation}
			\label{X}
			\frac{\partial R_{j \rightarrow k}}{\partial \mathbf{q}_n}=\frac{\frac{\partial\upsilon_{j,k}\left(\mathbf{q}_n\right)}{\partial \mathbf{q}_n}\left(\sum\limits_{\Omega_i>\Omega_k}\upsilon_{j,i}\left(\mathbf{q}_n\right)+\sigma^2\right)-\upsilon_{j,k}\left(\mathbf{q}_n\right)\frac{\partial\left(\sum\limits_{\Omega_i>\Omega_k}\upsilon_{j,i}\left(\mathbf{q}_n\right)\right)}{\partial \mathbf{q}_n}}{\text{ln}2\left(1+ \frac{ \upsilon_{j,k}\left(\mathbf{q}_n\right)}{\sum\limits_{\Omega_i>\Omega_k}\upsilon_{j,i}\left(\mathbf{q}_n\right)+\sigma^2}\right)\left(\sum\limits_{\Omega_i>\Omega_k}\upsilon_{j,i}\left(\mathbf{q}_n\right)+\sigma^2\right)^2}.
		\end{equation}
		\hrulefill
	\end{figure*}
	Note that the accuracy of approximation in \eqref{P11} is inversely dependent on $\zeta$, with higher precision achieved at lower $\zeta$. Now, the GDA can be adopted to effectively solve the unconstrained differential problem \eqref{P11}. 
	
	Afterwards, $\mathbf{\overline{Q}}^{\left(i\right)}$ and $R_k^{\text{min}^{\left(i\right)}}$ are updated by moving
	along the gradient direction in the $i$-th iteration, which can be written as
	\setcounter{equation}{31}
	\begin{subequations}
		\label{D}
		\begin{align}
			&\label{D1}R_k^{\text{min}^{\left(i+1\right)}}=R_k^{\text{min}^{\left(i\right)}}+\gamma^{\left(i\right)}\frac{\partial f\left( \mathbf{\overline{Q}}^{\left(i\right)},R_k^{\text{min}^{\left(i\right)}}\right)}{\partial R_k^{\text{min}^{\left(i\right)}}},
			\\
			&\label{D2}\begin{aligned}[t]
				\mathbf{\overline{Q}}^{(i+1)}&=\mathbf{\overline{Q}}^{(i)}+\gamma^{(i)} \frac{\partial f\left( \mathbf{\overline{Q}}^{\left(i\right)},R_k^{\text{min}^{\left(i\right)}}\right)}{\partial \mathbf{\overline{Q}}^{\left(i\right)}},
			\end{aligned}
		\end{align}
	\end{subequations}
	where $\gamma^{(i)}$ is the step size in the $i$-th iteration. Specifically, $\partial f\left( \mathbf{\overline{Q}},R_k^{\text{min}}\right)/\partial R_k^{\text{min}}$ is given by
	\begin{equation}
		\label{rkmin}
		\begin{aligned}
			\frac{\partial f\left( \mathbf{\overline{Q}},R_k^{\text{min}}\right)}{\partial R_k^{\text{min}}}=1-\kappa \frac{e^{p\left(R_k^{\text{min}}\right) / \zeta}}{1+e^{p\left(R_k^{\text{min}}\right) / \zeta}}.
		\end{aligned}
	\end{equation}
	Furthermore, based on the chain rule, the gradient of $\mathbf{\overline{Q}}$ is given by
	\begin{align}
		\label{dao2}
		\frac{\partial f\left( \mathbf{\overline{Q}},R_k^{\text{min}}\right)}{\partial \mathbf{\overline{Q}}}=\frac{\partial f\left( \mathbf{Q},R_k^{\text{min}}\right)}{\partial \mathbf{Q}} \circ A\frac{e^{-\mathbf{\overline{Q}}}}{\left(1+e^{\mathbf{-\overline{Q}}}\right)^2}\in \mathbb{C}^{2\times N},
	\end{align}
	where $\circ$ is the Hadamard multiplication. Note that the calculation of  $\partial f\left( \mathbf{Q},R_k^{\text{min}}\right)/\partial \mathbf{Q}$ is performed by computing the gradient w.r.t. $\mathbf{q}_n$, which constitutes the $n$-th column of the partial deviation matrix $\partial f\left( \mathbf{Q},R_k^{\text{min}}\right)/\partial \mathbf{Q}$. Specifically, $\partial f\left( \mathbf{Q},R_k^{\text{min}}\right)/\partial \mathbf{q}_n$ can be calculated as in \eqref{dao1}, shown at the top of this page, where $d\left(\mathbf{q}_n,\mathbf{q}_l\right)=D_{\text{min}}-\left\|\mathbf{q}_n-\mathbf{q}_l \right\|_2$ and  $\partial r\left(\mathbf{q}_n\right)/\partial \mathbf{q}_n$ is given by
	\setcounter{equation}{35}
	\begin{equation}
		\label{rdao}
		\frac{\partial r\left(\mathbf{q}_n\right)}{\partial \mathbf{q}_n}=\sum_{j =1}^{K}\sum_{\Omega_k<\Omega_i}\frac{\partial\upsilon_{j,i}\left(\mathbf{q}_n\right)}{\partial \mathbf{q}_n}-\frac{\partial\upsilon_{j,k}\left(\mathbf{q}_n\right)}{\partial \mathbf{q}_n},
	\end{equation}
	where $\upsilon_{j,k}\left(\mathbf{q}_n\right)=\left|  \mathbf{\omega}_{n,j} \mathbf{E}\left({\mathbf{q}_n}\right)\mathbf{\nu}_k+\alpha_{n,k}\right|^2$. Similarly, $\partial R_{j \rightarrow k}/\partial \mathbf{q}_n$ is given in \eqref{X}, shown at the top of this page, for which the detailed derivation is given in Appendix A. 
	\begin{algorithm}[t]
		\caption{GDA-based  algorithm for solving problem \eqref{P8}}
		\label{alg:pro3}
		\begin{algorithmic}[1]
			\STATE Initialize the MEs positions matrix $\mathbf{Q}$, the penalty factor $\kappa$, the smoothing parameter $\zeta$ and the initial step size $\gamma$;
			\STATE Set $\overline{\mathbf{Q}}^{(0)}=\mathbf{Q}$;
			\REPEAT
			\STATE Set the iteration index $i = 0$ for inner loop;
			\REPEAT
			\STATE Calculate $\partial f\left( \mathbf{\overline{Q}}^{\left(i\right)},R_k^{\text{min}^{\left(i\right)}}\right)/\partial R_k^{\text{min}^{\left(i\right)}}$ in \eqref{rkmin};
			\STATE Calculate $\partial r\left(\mathbf{q}_n^{(i)}\right)/\partial \mathbf{q}_n^{(i)}$ and $\partial R_{j \rightarrow k}/\partial \mathbf{q}_n^{(i)}$ in \eqref{rdao} and \eqref{X}, respectively;
			\STATE Calculate $\partial f\left( \mathbf{\overline{Q}}^{\left(i\right)},R_k^{\text{min}^{\left(i\right)}}\right)/\partial \mathbf{q}_n^{\left(i\right)}$ in \eqref{dao1};
			\STATE Calculate $\partial f\left( \mathbf{\overline{Q}}^{\left(i\right)},R_k^{\text{min}^{\left(i\right)}}\right)/\partial \mathbf{\overline{Q}}^{\left(i\right)}$ in \eqref{dao2};
			\STATE Update $\gamma^{(i)}=\iota_{\gamma}\gamma^{(i)}$ until the Armijo condition \eqref{armijo} is satisfied;
			\STATE Update $\mathbf{\overline{Q}}^{(i+1)}$ and $R_k^{\text{min}^{(i)}}$ in \eqref{D}
			;    \STATE $i \gets i + 1$
			\UNTIL{the increment of the objective value is less than the predefined threshold $\varepsilon_1$ or $i \geq i_{\text{max}}$; }
			\STATE Update $\zeta=l_{\zeta}\zeta$ and $\kappa=l_{\kappa}\kappa$;
			\UNTIL{the minimum distance constraint \eqref{P1_C8}, the resource allocation constraint \eqref{P1_C2}, and the minimum rate constraint \eqref{P15_C1} are satisfied;}
			\STATE \textbf{Output:} $ \mathbf{Q}=A\text{Sigmoid}\left(\overline{\mathbf{Q}}^{(i)}\right)$.
		\end{algorithmic}
	\end{algorithm}

	\vspace{-5pt}
	As for the determination of $\gamma$, we utilize the backtracking line search to ensure accuracy while maintaining low complexity. Specifically, $\gamma$ is set to be a relatively large value and is updated as $\gamma^{(i+1)}=\iota_{\gamma}\gamma^{(i)}$ with $0<l_{\gamma}<1$ until the following Armijo condition is satisfied:
	\setcounter{equation}{37}
	\begin{align}
		\label{armijo}
		f\left (\mathbf{\overline{Q}}^{(i+1)},R_k^{\text{min}^{\left( i+1\right)}}\right)&\geq f\left(\mathbf{\overline{Q}}^{\left(i\right)},R_k^{\text{min}^{\left( i\right)}}\right)\notag\\
		&+\mu\gamma^{(i)}\left\|  \nabla f\left(\mathbf{\overline{Q}}^{\left(i\right)},R_k^{\text{min}^{\left( i\right)}}\right)\right\|_F^2, 
	\end{align}
	where $0<\mu<1$ masters the increment range of the objective function. By imposing a lower bound on the acceptable increment in the objective function value, the Armijo condition dynamically balances the need for substantial progress per iteration against the risk of divergence, thus enhancing the robustness of the convergence process.

	It is noted that problem \eqref{P11} is equivalent to problem \eqref{P8} when $\zeta$ is small enough and $\kappa$ is sufficiently large~\cite{penalty}. Nevertheless, to maintain the balance between the primary objective and the penalty term, it is critical to avoid initializing $\kappa$ with an excessively high value, which could allow the penalty to dominate the optimization process. Therefore, we propose a nested loop iterative GDA framework, where $\partial f\left( \mathbf{\overline{Q}}^{(i)},R_k^{\text{min}^{(i)}}\right)/\partial R_k^{\text{min}^{(i)}}$ and $\partial f\left( \mathbf{\overline{Q}}^{(i)},R_k^{\text{min}^{(i)}}\right)/\partial \mathbf{\overline{Q}}^{(i)}$ in the $i$-th iteration are calculated in the inner loop. $\kappa$ and $\zeta$ are iteratively updated as $\kappa=l_{\kappa}\kappa,\quad l_{\kappa}>1$ and $\zeta=l_{\zeta}\zeta,\quad 0<l_{\zeta}<1$ in the outer loop until the minimum distance constraint \eqref{P1_C8}, the resource allocation constraint \eqref{P1_C2}, and the minimum rate constraint \eqref{P15_C1} are satisfied. The detailed algorithm is shown in \textbf{Algorithm~\ref{alg:pro3}}.
	
	\subsection{Convergence and Complexity Analysis of the Proposed AO Algorithm}
	Following the above three subsections, the details of the proposed AO algorithm to tackle the original problem \eqref{P1} are provided in \textbf{Algorithm~\ref{alg:pro4}}. Specifically, the iterative algorithm optimizes the BS beamforming vectors $\{ \mathbf{w}_k\}$, ME-STARS beamforming coefficients $\left\{ \mathbf{\Theta}_s\right\}$ and the positions of MEs $\left\{ \mathbf{Q}\right\}$ by solving the problem \eqref{P4}, \eqref{P7}, and \eqref{P9} respectively in each iteration. Each iteration  uses the solutions obtained from the previous iteration as the input local points except the first iteration. The iterative optimization process enables sequential refinement of the coupled variables, which enhances the overall system performance through cumulative improvements.
	\subsubsection{Convergence Analysis}
	In this section, we demonstrate the convergence of the proposed algorithm. With the initialized decoding order $\Omega$, the BS beamforming vectors $\left\{ \mathbf{w}_k  \right\}$, the ME-STARS beamforming coefficients $\left\{ \mathbf{\Theta}_r, \mathbf{\Theta}_t\right\}$, and the positions of MEs $\left\{\mathbf{Q}\right\}$, the following inequality can be obtained
	\setcounter{equation}{38}
	\begin{align}
		& R_{\text{sum}}\left(\mathbf{w}_k^{(\tau)},\mathbf{\Theta}_r^{(\tau)},\mathbf{\Theta}_t^{(\tau)},\mathbf{Q}^{(\tau)}\right)\notag\\
		&\stackrel{(a)}{\leqslant} R_{\text{sum}}\left(\mathbf{w}_k^{(\tau+1)},\mathbf{\Theta}_r^{(\tau)},\mathbf{\Theta}_t^{(\tau)},\mathbf{Q}^{(\tau)}\right)\notag\\
		&\stackrel{(b)}{\leqslant} R_{\text{sum}}\left(\mathbf{w}_k^{(\tau+1)},\mathbf{\Theta}_r^{(\tau+1)},\mathbf{\Theta}_t^{(\tau+1)},\mathbf{Q}^{(\tau)}\right)\notag\\
		&\stackrel{(c)}{\leqslant} R_{\text{sum}}\left(\mathbf{w}_k^{(\tau+1)},\mathbf{\Theta}_r^{(\tau+1)},\mathbf{\Theta}_t^{(\tau+1)},\mathbf{Q}^{(\tau+1)}\right),
	\end{align}
	where inequality $(a)$ holds for the property that the first-order Taylor expansions provide locally tight approximations with the fixed $\left\{ \mathbf{\Theta}_r^{(\tau)},\mathbf{\Theta}_t^{(\tau)} \right \}$ and $\left \{ \mathbf{Q}^{(\tau)} \right\}$; $(b)$ is valid since the solution to the reformulated problem \eqref{P7} consistently yields a lower bound on the objective value of the original problem \eqref{P5}; $(c)$ follows the fact that GDA-based algorithm can always optimize variables in the direction of increasing function values. Based on above analysis, it is obvious that the objective function value of \eqref{P1} is non-decreasing by alternatively solving problems \eqref{P3}, \eqref{P5}, and \eqref{P8} after each iteration. Therefore, since the objective function is upper bounded, this monotonicity property ensures convergence of 
	\textbf{Algorithm~\ref{alg:pro4}} to a stationary point of \eqref{P1}.
	
	\subsubsection{Complexity Analysis}
	It is noted that the complexity of \textbf{Algorithm~\ref{alg:pro4}} is mainly determined by that of \textbf{Algorithm~\ref{alg:pro1}}, \textbf{Algorithm~\ref{alg:pro2}} and \textbf{Algorithm~\ref{alg:pro3}}. If the interior point method is adopted to solve the standard convex problem, the computational complexities of the \textbf{Algorithm~\ref{alg:pro1}} and   \textbf{Algorithm~\ref{alg:pro2}} are $\mathcal{O}\left(I_1M^{3.5}\right)$ and $\mathcal{O}\left(I_2^{\text{in}}I_2^{\text{out}}N^{3.5}\right)$ respectively~\cite{inbook}, where $I_1$ is the number of iterations needed for the convergence for solving problem \eqref{P4} and $I_2^{\text{in}}$, $I_2^{\text{out}}$ represent the number of inner and outer iterations needed to reach the convergence, respectively. As for \textbf{Algorithm~\ref{alg:pro3}}, the computational complexity mainly depends on the calculation of the gradient, which is $\mathcal{O}\left(\left(N+K\right)^2L_{\text{B,S}}L_{\text{S},k}\right)$. The computational complexity of the backtracking line search is $\mathcal{O}\left(I_{\text{ls}}MN\right)$, where $I_{\text{ls}}$ is the number of iterations required to determine the accurate step size. Thus, \textbf{Algorithm~\ref{alg:pro3}} exhibits a complexity of $\mathcal{O}\left(I_3^{\text{out}}I_3^{\text{in}}\left(\left(N+K\right)^2L_{\text{B,S}}L_{\text{S},k}+I_{\text{ls}}MN\right)\right)$, where $I_3^{\text{out}}$ and $I_3^{\text{in}}$ represent the number of outer and inner iterations. Therefore, since the computational complexity of \textbf{Algorithm 1, 2, 3} are polynomial in $M$, $N$ and $K$, the practical implementation of \textbf{Algorithm 4} is feasible.
	\begin{algorithm}
		\caption{AO algorithm for solving problem \eqref{P1}}
		\label{alg:pro4}
		\begin{algorithmic}[1]
			\STATE Initialize feasible points $ \left\{\mathbf{\Theta}_{r}^{(0)},\mathbf{\Theta}_{t}^{(0)}\right\}$, $\left\{ \mathbf{w}_k^{(0)} \right\}$ and $\left\{  \mathbf{Q}^{(0)} \right\}$;
			\STATE Initialize the decoding order $\Omega$;
			\STATE Set the iteration index $\tau = 0$;
			\REPEAT
			\STATE Update the BS beamforming vectors $\left\{ \mathbf{w}_k^{(\tau+1)}\right\}$ for given $ \left\{\mathbf{\Theta}_{r}^{(\tau)},\mathbf{\Theta}_{t}^{(\tau)}\right\}$ and $\left\{  \mathbf{Q}^{(\tau)}\right\}$by invoking  \textbf{Algorithm~\ref{alg:pro1}};
			\STATE Update the ME-STARS beamforming coefficients $ \left\{\mathbf{\Theta}_{r}^{(\tau+1)},\mathbf{\Theta}_{t}^{(\tau+1)}\right\}$ for given $\left\{ \mathbf{w}_k^{(\tau)}\right\}$ and $\left\{  \mathbf{Q}^{(\tau)}\right\}$by invoking  \textbf{Algorithm~\ref{alg:pro2}};
			\STATE Update the positions of MEs $\left\{  \mathbf{Q}^{(\tau+1)}\right\}$ for given $ \left\{\mathbf{\Theta}_{r}^{(\tau)},\mathbf{\Theta}_{t}^{(\tau)}\right\}$ and $\left\{ \mathbf{w}_k^{(\tau)}\right\}$ by invoking  \textbf{Algorithm~\ref{alg:pro3}};
			\STATE $\tau$=$\tau$+1;
			\UNTIL{the fractional increase of the objective value is less than $\varepsilon_1$ or $\tau\geq\tau^{\text{max}}$;}
			\STATE Output: $\left\{\mathbf{w}_k^*\right\}$, $ \left\{\mathbf{\Theta}_{r}^*,\mathbf{\Theta}_{t}^*\right\}$ and $\left\{  \mathbf{Q}^*\right\}$.
		\end{algorithmic}
	\end{algorithm}
	\section{Numerical Results}
    This section presents numerical results to evaluate the performance of the proposed algorithms. Under considered 3D coordinate system, the location of the BS is set at $\left(20,0,10\right)$ meters, while the users in the transmission and reflection region are randomly distributed within the circle areas centered at $\left(20,15,0\right)$ and $\left(20,5,0\right)$ with the radius of 2 meters, respectively. The number of paths in the FRV channel model is assumed to be the same\footnote{This assumption is to simplify the simulation setup, while the proposed algorithms are applicable to various values of $L_{\text{B,S}}$ and $L_{\text{S},k}, \forall k \in \mathcal{K}$.}, i.e., $L_{\text{B,S}}=L_{\text{S},k}=L,\forall k \in \mathcal{K}$. The path response matrix $\boldsymbol{\Sigma}_{\text{B,S}}$ is a diagonal matrix with circularly symmetric complex Gaussian (CSCG)-distributed entries $\mathcal{CN}\left(0,\beta_{\text{B,S}}/L\right)$, where $\beta_{\text{B,S}}=\beta_0d_{\text{B,S}}^{-\alpha_0}$  denotes the expected BS-STARS channel power gain, with $\beta_0$,  $d_{\text{B,S}}$, and $\alpha_0$ representing the channel power gain at 1 m reference distance, the BS-STARS distance, and the path-loss exponent, respectively. The diagonal entries of $\boldsymbol{\Sigma}_{k}$ follow the CSCG distribution $\mathcal{CN}(0,\beta_k/L)$, where $\beta_{\text{S},k}=\beta_0d_{\text{S},k}^{-\alpha_0}$ is the expected STARS-user channel power gain with $d_{\text{S},k}$ representing the distance between the ME-STARS and the user $k$. The AODs and AOAs involved in the FRV model $\theta_{\text{B}}^e$, $\varphi_{\text{B}}^e$, $\theta_{\text{S,in}}^f$, $\varphi_{\text{S,in}}^f$, $\theta_{\text{S},k}^g$, $\varphi_{\text{S},k}^g$, $\forall e,f,g$ are assumed to be independently obtained by randomly generating within the range $\left[-\frac{\pi}{2},\frac{\pi}{2}\right]$. Unless otherwise stated, all parameter values are given in Table \ref{tab:parameters}. Considering that the proposed AO-based algorithm may converge to different local solutions, we perform multiple initializations of the optimization variables and select the one yielding the best performance. It is noted that all the points in Fig. \ref{fig2} - Fig. \ref{fig5} are obtained through $10^3$ user distributions and channel realizations.
	\begin{table}[htbp]
		\centering
		\caption{Simulation Parameters.}
		\label{tab:parameters}
		\begin{tabular}{cccc}
			\toprule
			Parameter Description & Value  \\
			\midrule
			Carrier length $\lambda$ &0.1m      \\
			Number of antennas at the BS $M$& 2       \\
			Number of MEs $N$ & 8      \\
			Number of users $K$ & 4      \\
			Path-loss exponent $\alpha_0$ & 2.2\\
			Channel power gain at 1 meter $\beta_0$    & -30 dB\\
			BS maximum transmit power $P_{\text{B}}$   & 30 dBm \\
			Noise power $\sigma^2$   &-90 dBm\\
			Number of paths $L$ & 2      \\
			Length of the sides of moving region $A$ &4 $\lambda$ \\
			
			Initial gradient-ascent step size  $\gamma^{(0)}$ & 10\\
			Step size shrinking parameter  $l_{\gamma}$ & 0.9 \\
			Initial penalty factor $\kappa$ in Algorithm 2 $\&$ 3 & $10^{-4}$  \\
			Initial smoothing  factor $\zeta$ in Algorithm 3 & $10^{-4}$  \\
			Penalty scaling factor $l_{\kappa}$  in Algorithm 2 $\&$ 3    & 10 \\
			Smoothing scaling factor $l_{\zeta}$ in Algorithm  3     & 0.1 \\
			
			\bottomrule
		\end{tabular}    
	\end{table}
	
	To assess the performance advantages of the ME-STARS-assisted NOMA system and validate the proposed algorithm, we consider the following benchmark schemes:
	\begin{itemize}
		\item \textbf{FPE-STARS}: In this case, the $N$ FPEs of the STARS are assumed to be distributed uniformly at intervals of $\frac{\lambda}{2}$.
		The sum rate maximization problem can be resolved by iteratively solving the BS beamforming and ME-STARS beamforming subproblems with the algorithms proposed in Section III. 
		\item \textbf{RPE-STARS (STARS with random-position elements)}: In this case, the $N$ elements of the STARS are assumed to be distributed randomly within the given moving region $\mathcal{C}_{\text{ME}}$. Specifically, to guarantee the minimum distance condition for fair comparison, we fill the moving region $\mathcal{C}_{\rm ME}$ with circles of diameter $\frac{\lambda}{2}$, and then randomly select the centers of $N$ circles as the element positions. The sum rate maximization problem can be addressed via the same algorithm as that for the FPE-STARS.
		\item \textbf{ME-RIS}: This benchmark employs a dual-RIS architecture to achieve full-space coverage, comprising one RIS only for reflection and one RIS only for transmission deployed at the location of the ME-STARS. To ensure fairness, each RIS is set to possess $N/2$ elements.
		\item \textbf{FPE-RIS}: This case adopts the same setup as the ME-RIS, except that elements positions are fixed with the spacing of $\frac{\lambda}{2}$.
	\end{itemize}
	\begin{figure}[t]
		\centering
		\includegraphics[scale=0.32]{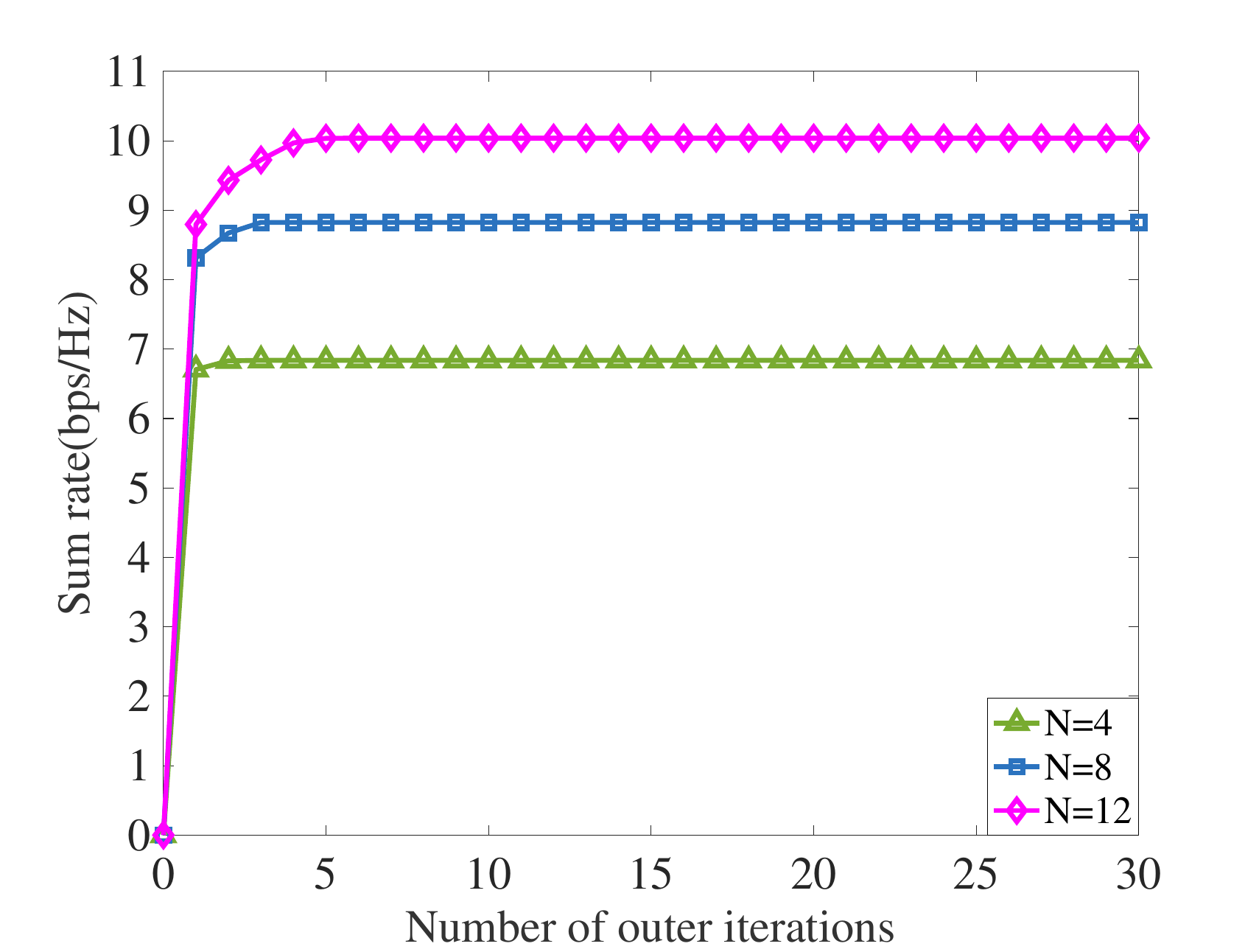}
		\caption{Coverage performance versus the number of ME-STARS elements. }
		\label{fig1}
	\end{figure}
	
	Fig. $\ref{fig1}$ reveals the convergence behavior of the  overall AO algorithm for different number of ME-STARS elements $N$. The curves are derived from a single stochastic channel realization.  It can be observed that the sum rate increases rapidly at the beginning of the iterations and becomes stable within 5 outer iterations under the considered simulation settings. Also, the number of iterations necessary to achieve convergence increases with the increment of the number of MEs. This is expected as larger number of MEs leads to more complicated joint beamforming and MEs positions design, thus resulting in increased computation burden and lower convergence rate.
	
	\begin{figure}[t]
		\centering
		\includegraphics[scale=0.32]{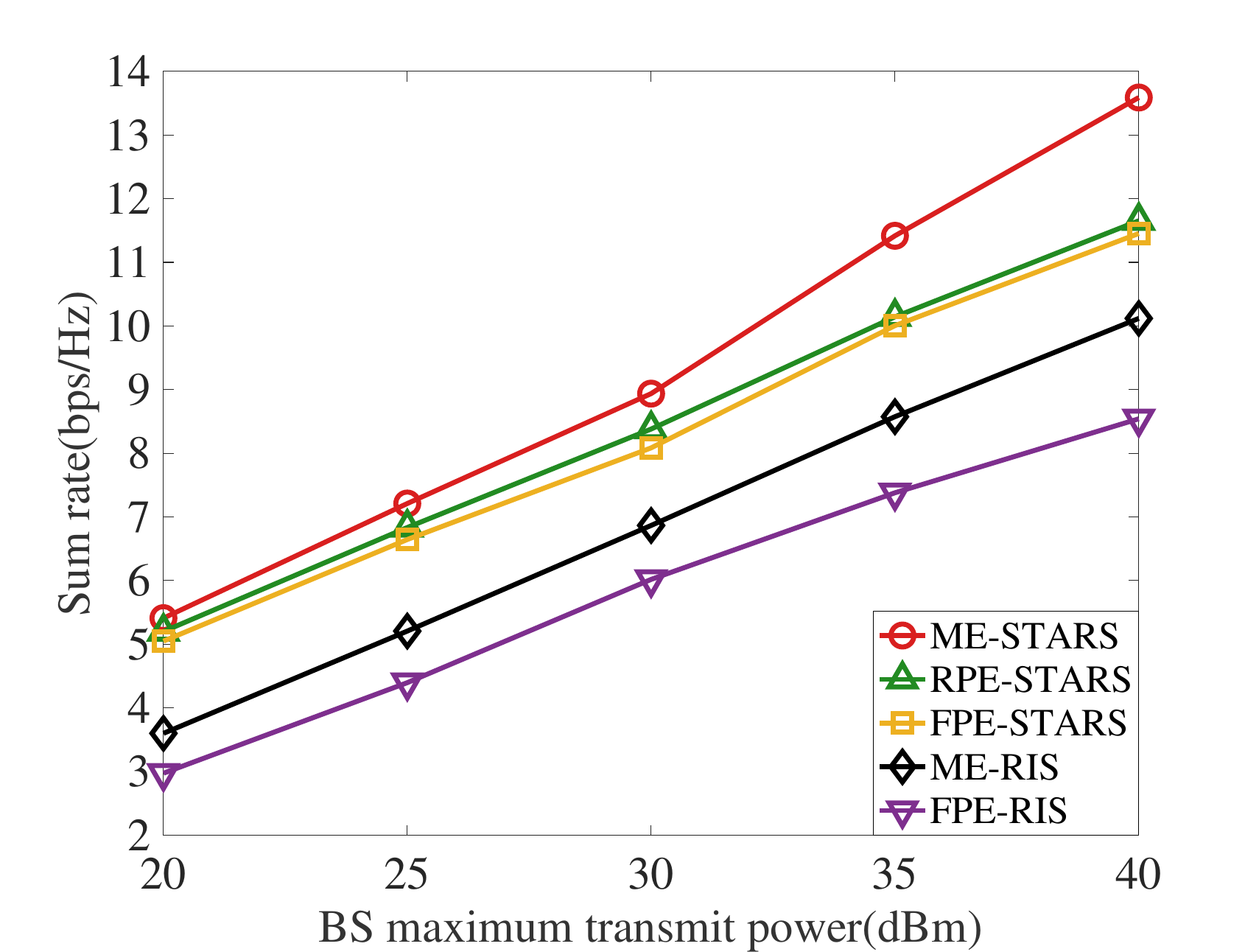}
		\caption{Sum rate versus the BS maximum transmit power.}
		\label{fig2}
	\end{figure}
	
	In Fig. \ref{fig2}, we investigate the sum achievable rate of different schemes versus the BS maximum transmit power $P_{\text{B}}$. It is shown that all baseline schemes increase as $P_{\text{B}}$ increases. It is also observed that ME-STARS outperforms the RPE-STARS and FPE-STARS by around $8.81\%$ and $11.25\%$ on average, respectively, which provides support for the underlying hypothesis that the performance gain of ME-STARS stems from its ability of utilizing spatial-domain diversity. Moreover, the performance improvement of ME-STARS compared to FPE-STARS and RPE-STARS exhibits a higher growth rate with increased $P_{\text{B}}$. This can be explained as follows. In the interference-limited regime, the ME-STARS possesses higher spatial DoFs for moving MEs to positions with reduced channel correlations among users, thereby achieving higher interference mitigation gain.
	
	
	\begin{figure}[t]
		\centering
		\includegraphics[scale=0.32]{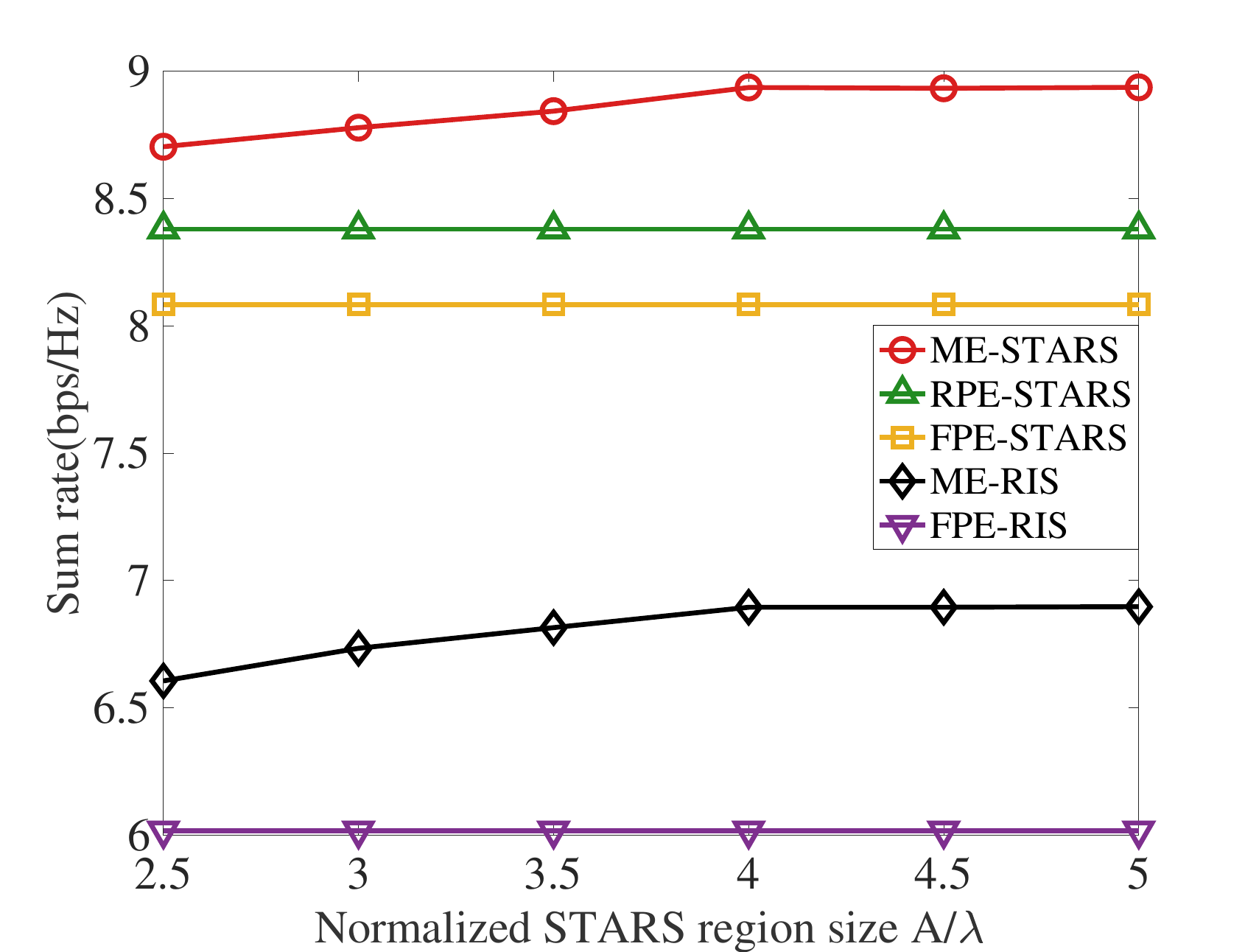}
		\caption{Sum rate versus the normalized STARS region size.}
		\label{fig3}
	\end{figure}
	
	In Fig. \ref{fig3}, we study the sum achievable rate versus the normalized moving region size $A/\lambda$. Note that the sum achievable rate across all ME schemes exhibits enhancement as the ME region size is expanded. It can be observed that the performance of ME-STARS is $8.65\%$ and $10.60\%$ higher than that of FPE-STARS when normalized region size $A/\lambda$ is 3 and 4 respectively. The reason behind this is that, with the increment of the moving region size, higher spatial-domain diversity gain brought by extra DoFs can be explored, and MEs can achieve more flexible position adjustment for constructing preferable channel conditions. It is noted that the sum rate improvement gradually converges with the increment of $A/\lambda$. This can be attributed to the fact that, beyond a certain threshold, the number of optimal antenna positions exceeds that of deployable antennas, making additional spatial expansion redundant in terms of performance gain. In other words, it implies that maximal sum rate can be obtained within a finite moving region size. 
	
	
	\begin{figure}[t]
		\centering
		\includegraphics[scale=0.32]{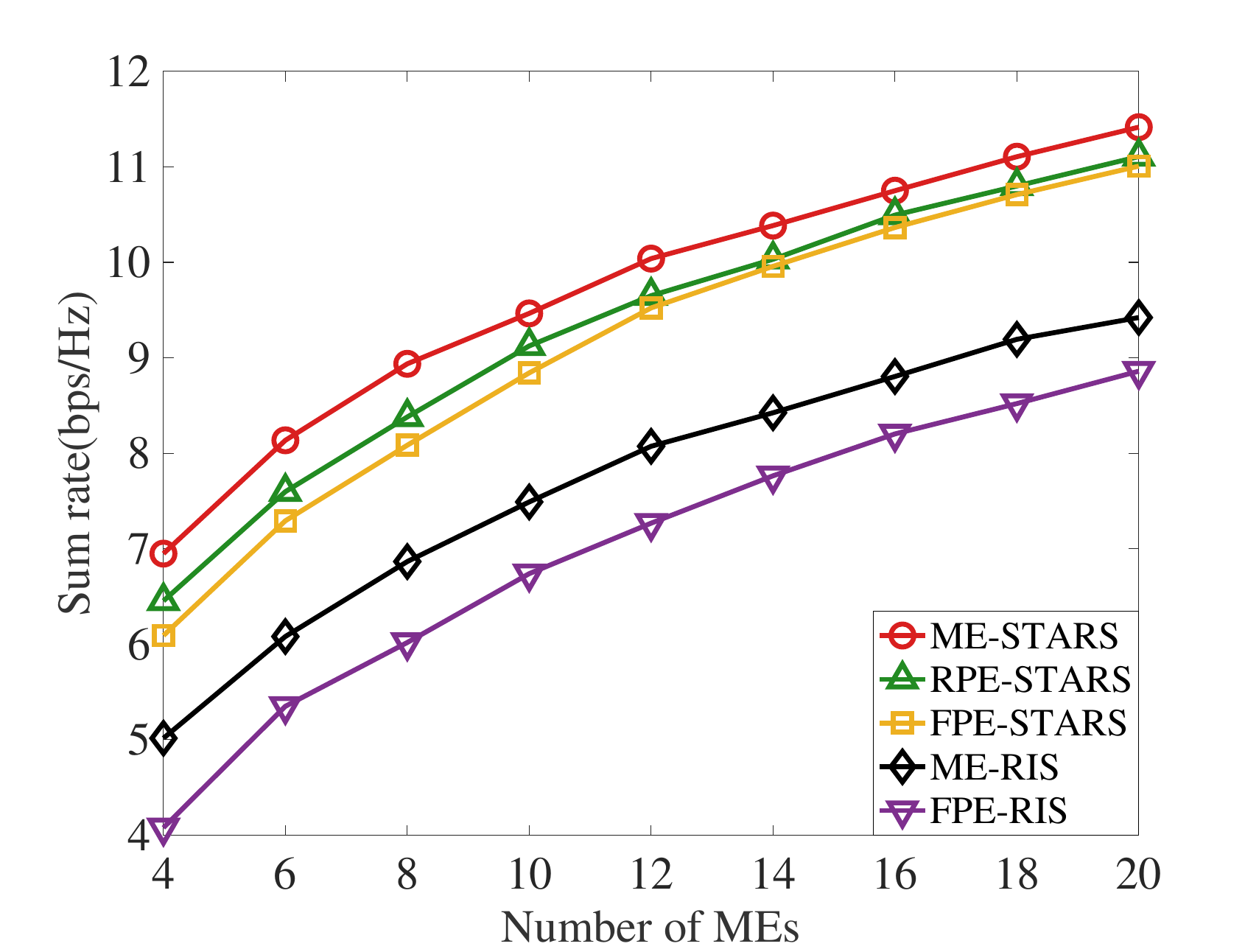}
		\caption{Sum rate versus the number of elements.}
		\label{fig4}
	\end{figure}
	\begin{figure}[t]
		\centering
		\includegraphics[scale=0.32]{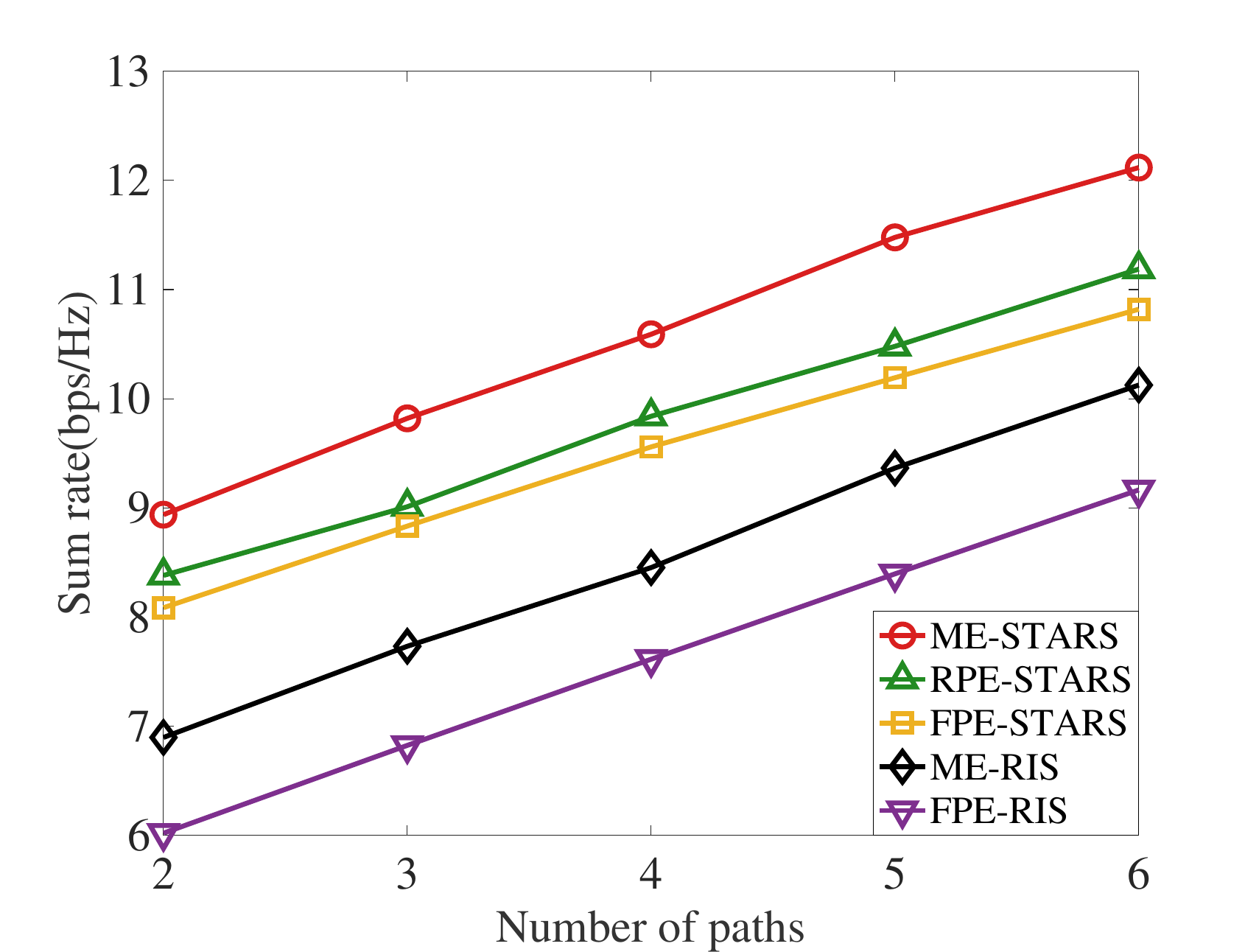}
		\caption{Sum rate versus the number of paths of each channel.}
		\label{fig5}
	\end{figure}
	In Fig. \ref{fig4}, we investigate the sum achievable rate versus the number of ME elements $N$. All schemes exhibit a monotonically increasing in terms of sum rates with growing $N$, while ME-STARS maintains superior performance. This expected behavior stems from the enhanced passive beamforming gain and improved channel conditions enabled by larger $N$ values. It is also noted that the ME-STARS performance advantage rate 
	over FPE-STARS gets smaller with the increment of $N$. This is because, the number of locally optimal positions that maximize the sum rate is inherently limited due to the constrained MEs moving region size. Thus, as the number of MEs $N$ increases beyond the number of local maxima, it becomes impossible to allocate all MEs to their preferable positions. Consequently, this limitation results in a saturation of the achievable sum rate.
	
	Fig. \ref{fig5} illustrates the sum achievable rate performance versus the number of paths $L$. Note that the sum rate of all schemes increases with the increment of $L$. It is also shown that ME-STARS achieves the highest performance under all the number of $L$. The underlying rationale for this phenomenon can be analyzed from two perspectives. First, as the number of channel paths $L$ increases, the multiuser channel matrix exhibits a higher rank, indicating reduced correlation among user channel vectors. This characteristic facilitates more effective multiuser interference suppression in the system. Second, the presence of additional multipath components leads to more pronounced small-scale fading effects, which in turn generates a larger number of local maxima within the defined moving region. Consequently, the MEs-based schemes achieve improved performance through access to a wider set of favorable local optimal positions.
	
	\section{Conclusion}
	This paper investigated a novel ME-STARS-assisted NOMA communications, formulating a sum achievable rate maximization problem through joint optimization of BS beamforming as well as the passive beamforming and MEs positions. An iterative AO-based algorithm framework was proposed to decompose the highly-coupled non-convex problem into three tractable subproblems. Specifically, for the BS and ME-STARS beamforming optimization subproblems, the SCA technique was invoked. For the MEs positions optimization subproblem, the GDA was proposed to iteratively adjust MEs positions within the confined region, where the objective function was integrated with the penalty terms. Simulation results demonstrated the effectiveness of the proposed algorithm and the superiority of the proposed ME-STARS-assisted NOMA communications. The performance gain in terms of sum rate compared to the FPE-STARS
	indicated that the dynamically adjustable MEs can further exploit the extra spatial DoFs. Moreover, robust algorithm design is an important topic for future research to mitigate imperfect CSI effects, particularly in dynamic environments where channel estimation errors can degrade system performance.
	
	\begin{appendices} 
		\section{DETAILED DERIVATION OF $\partial f\left( \mathbf{Q},R_k^{\rm min}\right)/\partial \mathbf{q}_n$}
		
		\begin{figure*}[t]
			
			\begin{align}
				\label{L2}
				\upsilon_{j,k}\left( \mathbf{q}_n \right)=&\left|    \sum_{a=1}^{L_{\text{B,S}}} \sum_{b=1}^{L_{S,j}}\left|\omega_{n, j}^a\right|\left|\nu_{k}^b\right| e^{j\left(\frac{2\pi}{\lambda}\left(\rho^a_{\text{S,}k}\left(\mathbf{q}_{n}\right)-\rho^b_{\text{S,in}}\left(\mathbf{q}_{n}\right)\right)+\angle \omega_{n, j}^a+\angle \nu_{k}^b\right)}  +\left|\alpha_{n, k}\right| e^{j \angle \alpha_{n, k}} \right|^2 \notag\\       =&\left(\sum_{a=1}^{L_{\text{B,S}}} \sum_{b=1}^{L_{S,j}}\left|\omega_{n, j}^a\right|\left|\nu_{k}^b\right| \cos\left(\frac{2 \pi}{\lambda}\left(\chi_{a,b} x_n+\psi_{a,b} y_n\right)+\angle \nu_j^a+\angle \omega_{n, j}^b\right)   +\left|\alpha_{n, k}\right| \cos \angle \alpha_{n, k} \right)^2\notag\\
				+&\left(\sum_{a=1}^{L_{\text{B,S}}} \sum_{b=1}^{L_{S,j}}\left|\omega_{n, j}^a\right|\left|\nu_{k}^b\right| \sin\left(\frac{2 \pi}{\lambda}\left(\chi_{a,b} x_n+\psi_{a,b}y_n\right)+\angle \nu_j^a+\angle \omega_{n, j}^b\right)   +\left|\alpha_{n, k}\right| \sin \angle \alpha_{n, k} \right)^2,
				\tag{A.1}
			\end{align}
			\hrulefill
		\end{figure*}
		\begin{figure*}[b]
			\vspace*{0pt} 
			\hrulefill
			\begin{align}
				\label{up}
				c_{j,k}^{a,b}\left(\mathbf{q}_n\right) =\sin \left(\angle \alpha_{n, k}-\left(\frac{2 \pi}{\lambda}\left(\chi_{a,b} x_n+\psi_{a,b} y_n\right)+\angle \nu_k^b+\angle \omega_{n, j}^a\right)\right).
				\tag{A.3}
			\end{align}
			\setcounter{equation}{42}
			\begin{align}
				\label{Xdao2}
				\tag{A.5}        \frac{\partial\left(\sum\limits_{\Omega_i>\Omega_k}\upsilon_{j,i}\left(\mathbf{q}_n\right)\right)}{\partial x_n}=\frac{4 \pi\left|\alpha_{n, i}\right|}{\lambda} \sum\limits_{\Omega_i>\Omega_k}\sum_{a=1}^{L_{\text{B,S}}} \sum_{b=1}^{L_{\text{S}, j}} \chi_{a,b}\left|\omega_{n, j}^a\right|\left|\nu_i^b\right|c_{j,i}^{a,b}\left(\mathbf{q}_n\right).
			\end{align}
			\begin{align}
				\label{Y2}
				\tag{A.6}\frac{\partial\left(\sum\limits_{\Omega_i>\Omega_k}\upsilon_{j,i}\left(\mathbf{q}_n\right)\right)}{\partial y_n}=\frac{4 \pi\left|\alpha_{n, i}\right|}{\lambda} \sum\limits_{\Omega_i>\Omega_k}\sum_{a=1}^{L_{\text{B,S}}} \sum_{b=1}^{L_{\text{S}, j}} \psi_{a,b}\left|\omega_{n, j}^a\right|\left|\nu_i^b\right|c_{j,i}^{a,b}\left(\mathbf{q}_n\right).
			\end{align}
		\end{figure*}
		By denoting $\omega_{n,j}=\left |\omega_{n,j}\right|e^{j \angle \omega_{n,j}}$, $\nu_k=|\nu_k|e^{j \angle \nu_k}$ and $\alpha_{n,k}=\left |\alpha_{n,k}\right|e^{j \angle \alpha_{n,k}}$, $\upsilon_{j,k}\left( \mathbf{q}_n \right)$ can be expressed as in \eqref{L2}, shown at the top of this page, where $\chi_{a,b}=\cos\theta^a_{\text{S},k} \sin\varphi^a_{\text{S},k}-\cos \theta^b_{\text{S,in}} \sin \varphi^b_{\text{S,in}}$ and $\psi_{a,b}=\sin\theta^a_{\text{S},k}-\sin \theta^b_{\text{S,in}}$ are the multiplication factors of  $x_n$  and $y_n$. 
		
		It is obvious that the key to obtain both $\partial r\left(\mathbf{q}_n\right)/\partial \mathbf{q}_n$ and $\partial R_{j \rightarrow k}/\partial \mathbf{q}_n$ is the calculation of $\partial\upsilon_{j,k}\left(\mathbf{q}_n\right)/\partial \mathbf{q}_n$, which is a vector composed of the partial deviation of $\upsilon_{j,k}\left( \mathbf{q}_n \right)$ w.r.t $x_n$ and $y_n$, i.e., $\nabla \upsilon_{j,k}\left( \mathbf{q}_n\right)=\left[ \frac{\partial\upsilon_{j,k}\left( \mathbf{q}_n\right)}{\partial x_n} \frac{\partial\upsilon_{j,k}\left( \mathbf{q}_n\right)}{\partial y_n}\right]^T$. Subsequently,  $\partial\upsilon_{j,k}\left( \mathbf{q}_n \right)/\partial x_n$ can be given as
		
		\begin{align}
			\label{Xdao1}
			\frac{\partial\upsilon_{j,k}\left( \mathbf{q}_n \right)}{\partial x_n}=\frac{4 \pi\left|\alpha_{n, k}\right|}{\lambda} \sum_{a=1}^{L_{\text{B,S}}} \sum_{b=1}^{L_{\text{S}, j}} \chi_{a,b}\left|\omega_{n, j}^a\right|\left|\nu_k^b\right|c_{j,k}^{a,b}\left(\mathbf{q}_n\right),
			\tag{A.2}
		\end{align}
		where $c_{j,k}^{a,b}\left(\mathbf{q}_n\right)$ is given in \eqref{up}. Similarly, $\upsilon_{j,k}\left( \mathbf{q}_n \right)/\partial y_n$ can be given as
		\setcounter{equation}{41}
		\begin{align}
			\label{Y1}
			\tag{A.4}\frac{\partial\upsilon_{j,k}\left( \mathbf{q}_n \right)}{\partial y_n}=\frac{4 \pi\left|\alpha_{n, k}\right|}{\lambda} \sum_{a=1}^{L_{\text{B,S}}} \sum_{b=1}^{L_{\text{S}, j}} \psi_{a,b}\left|\omega_{n, j}^a\right|\left|\nu_k^b\right|c_{j,k}^{a,b}\left(\mathbf{q}_n\right).
		\end{align}
		By combing \eqref{Xdao1} and \eqref{Y1}, we can easily obtain $\partial r\left(\mathbf{q}_n\right)/\partial \mathbf{q}_n$. Accordingly, $\partial\left(\sum\limits_{\Omega_i>\Omega_k}\upsilon_{j,i}\left(\mathbf{q}_n\right)\right)/\partial \mathbf{q}_n$ can be calculated as in \eqref{Xdao2} and \eqref{Y2}, shown at the top of this page. Based on the above equations, the derivation of the gradient vector of $f\left( \mathbf{Q},R_k^{\text{min}}\right)$ at the point $\mathbf{q}_n$ can be obtained. 
	\end{appendices}
	
	\bibliographystyle{ieeetr}
	\bibliography{mybib}
\end{document}